\documentclass[letterpaper]{article}

\PassOptionsToPackage{numbers, compress}{natbib}

\usepackage[final]{neurips_2021}

\usepackage[utf8]{inputenc} \usepackage[T1]{fontenc}    \usepackage{hyperref}       \usepackage{url}            \usepackage{booktabs}       \usepackage{amsfonts}       \usepackage{nicefrac}       \usepackage{microtype}      \usepackage{xcolor}         \usepackage{amsmath,amsthm,amssymb}
\usepackage{xspace}
\usepackage{algorithm}
\usepackage{algpseudocode}
\usepackage{graphicx}
\usepackage{epstopdf}
\usepackage{caption}
\usepackage{subcaption}
\usepackage{paralist}
\usepackage[page]{appendix}

\hypersetup{colorlinks=true,allcolors=blue}

\theoremstyle{plain}
\newtheorem{theorem}{Theorem}[section]
\newtheorem{lemma}[theorem]{Lemma}

\newtheorem{corollary}[theorem]{Corollary}

\newtheorem{fact}[theorem]{Fact}

\theoremstyle{definition}
\newtheorem{definition}{Definition}[section]

\theoremstyle{remark}

\newcommand{\ProblemName}[1]{\textsc{#1}}
\newcommand{\kzC}{\ProblemName{$(k, z)$-Clustering}\xspace}
\newcommand{\kMedian}{\ProblemName{$k$-Median}\xspace}
\newcommand{\kMeans}{\ProblemName{$k$-Means}\xspace}
\newcommand{\kMeanspp}{\ProblemName{$k$-Means++}\xspace}
\newcommand{\kCenter}{\ProblemName{$k$-Center}\xspace}

\newcommand{\RR}{{\mathbb R}}

\DeclareMathOperator{\sdim}{sdim}
\DeclareMathOperator{\dist}{dist}
\DeclareMathOperator{\cost}{cost}
\DeclareMathOperator{\poly}{poly}

\DeclareMathOperator{\proj}{proj}

\providecommand{\set}[1]{{\{#1\}}}

\title{Coresets for Clustering with Missing Values}

\author{Vladimir Braverman \\
  Johns Hopkins University \\
  \texttt{vova@cs.jhu.edu} \\
  \And
  Shaofeng H.-C. Jiang \\
  Peking University \\
  \texttt{shaofeng.jiang@pku.edu.cn} \\
  \And
  Robert Krauthgamer \\
  Weizmann Institute of Science \\
  \texttt{robert.krauthgamer@weizmann.ac.il} \\
  \And
  Xuan Wu \\
  Johns Hopkins University \\
  \texttt{wu3412790@gmail.com}
}

\begin{document}

\maketitle
\begin{abstract}
We provide the first coreset for clustering points in $\mathbb{R}^d$ that have multiple missing values (coordinates).
Previous coreset constructions only allow one missing coordinate.
The challenge in this setting is that objective functions, like \kMeans,
are evaluated only on the set of available (non-missing) coordinates,
which varies across points.
Recall that an $\epsilon$-coreset of a large dataset
is a small proxy, usually a reweighted subset of points,
that $(1+\epsilon)$-approximates the clustering objective
for every possible center set.

Our coresets for \kMeans and \kMedian clustering have size
$(jk)^{O(\min(j,k))} (\epsilon^{-1} d \log n)^2$,
where $n$ is the number of data points, $d$ is the dimension and $j$ is the maximum number of missing coordinates for each data point.
We further design an algorithm to construct these coresets in near-linear time,
and consequently improve a recent quadratic-time PTAS
for \kMeans with missing values [Eiben et al., SODA 2021]
to near-linear time.

We validate our coreset construction,
which is based on importance sampling and is easy to implement,
on various real data sets.
Our coreset exhibits a flexible tradeoff between coreset size and accuracy,
and generally outperforms the uniform-sampling baseline.
Furthermore, it significantly speeds up a Lloyd's-style heuristic for \kMeans with missing values.
\end{abstract}

 \section{Introduction}
\label{sec:intro}
We consider coresets and approximation algorithms for $k$-clustering problems, particularly \kMeans \footnote{In the usual \kMeans problem (without missing coordinates),
the input is a data set $X \subset \mathbb{R}^d$ and the goal is to find a center set $C \subset \mathbb{R}^d, |C| = k$ that minimizes the sum of squared distances from every $x \in X$ to $C$.}
and more generally \kzC (see Definition~\ref{def:kzc}),
for points in $\mathbb{R}^d$ with \emph{missing values (coordinates)}.
The presence of missing values in data sets is a common phenomenon,
and dealing with it is a fundamental challenge in data science.
While data imputation is a very popular method for handling missing values,
it often requires prior knowledge which might not be available,
or statistical assumptions on the missing values that might be difficult to verify~\cite{allison2001missing,little2019statistical}.
In contrast, our worst-case approach does not requires any prior knowledge.
Specifically, in our context of clustering, the distance $\dist(x, c)$
between a clustering center point $c$ and a data point $x$ is evaluated only
on the available (i.e., non-missing) coordinates.
Similar models that aim to minimize clustering costs using only the available
coordinates have been proposed in previous work~\cite{hathaway2001fuzzy,wagstaff2004clustering,CCB16,8693952},
and some other relevant works were discussed in a survey~\cite{DBLP:conf/icdim/HimmelspachC10}.

Clustering under this distance function,
which is evaluated only on the available coordinates,
is a formidable computational challenge,
because distances do not satisfy the triangle inequality,
and therefore many classical and effective clustering algorithms,
such as \kMeanspp~\cite{DBLP:conf/soda/ArthurV07},
cannot be readily applied or even be defined properly.
Despite the algorithmic interest in clustering with missing values,
the problem is still not well understood and only a few results are known.
In a pioneering work, Gao, Langberg and Schulman~\cite{DBLP:journals/dcg/GaoLS08}
initiated the algorithmic study of
the \kCenter problem with missing values.
They took a geometric perspective and interpreted the \kCenter with missing values problem as an affine-subspace clustering problem,
and followup work~\cite{DBLP:journals/talg/GaoLS10,DBLP:conf/soda/LeeS13}
has subsequently improved and generalized their algorithm.
Only very recently, approximation algorithms for objectives other than \kCenter, particularly \kMeans, were obtained for the limited case of at most one missing coordinate in each input point~\cite{DBLP:conf/nips/MaromF19}
or for constant number of missing coordinates~\cite{DBLP:conf/soda/EibenFGLPS21}.

We focus on designing coresets for clustering with missing values.
Roughly speaking, an $\epsilon$-coreset is a small proxy of the data set,
such that the clustering objective is preserved within $(1\pm \epsilon)$ factor for all center sets (see Definition~\ref{def:coreset} for formal definition).
Efficient constructions of small $\epsilon$-coresets usually lead to
efficient approximations schemes, since the input size is reduced to that of the coreset, see e.g.~\cite{huang2018epsilon,DBLP:journals/siamcomp/FriggstadRS19,DBLP:conf/nips/MaromF19}.
Moreover, apart from speeding up approximation algorithms in the classical setting (offline computation),
coresets can also be applied to design streaming~\cite{DBLP:conf/stoc/Har-PeledM04,DBLP:conf/stoc/FrahlingS05,DBLP:conf/icml/BravermanFLSY17},
distributed~\cite{DBLP:conf/nips/BalcanEL13,DBLP:conf/uai/ReddiPS15,DBLP:conf/kdd/BachemL018},
and dynamic algorithms~\cite{DBLP:journals/dcg/Chan09,DBLP:conf/esa/HenzingerK20},
which are effective methods/models for dealing with big data,
and recently coresets were used even in neural networks~\cite{DBLP:conf/iclr/MussayOBZF20}.

\subsection{Our Results}
\label{sec:result}

\paragraph{Coresets.}
Our main result, stated in Theorem~\ref{thm:main_informal},
is a near-linear time construction of coresets for \kMeans with missing values.
Here, an $\epsilon$-coreset for \kMeans
for a data set $X$ in $\mathbb{R}^d$ with missing coordinates
is a weighted subset $S\subseteq X$ with weights $w : S \to \mathbb{R}_+$,
such that
\begin{align*}
    \forall C \subset \mathbb{R}^d, |C| = k, \qquad
    \sum_{x \in S}{w(x) \cdot \dist^2(x, C)} \in (1 \pm \epsilon)
    \sum_{x \in X}{\dist^2(x, C)} ,
\end{align*}
where $\dist(x, c) := \sqrt{ \sum_{i: \text{$x_i$ not missing}}{ (x_i - c_i)^2 } }$,
and $\dist(x, C) := \min_{c \in C}{\dist(x, c)}$;
note that the center set $C$ does not contain missing values.
More generally, our coreset also works for \kzC, which includes \kMedian (see Definition~\ref{def:kzc} and Definition~\ref{def:coreset}).
Throughout, we use $\tilde{O}(f)$ to denote $O(f \poly \log f)$.

\begin{theorem}[Informal version of Theorem~\ref{thm:main_full}]
    \label{thm:main_informal}
    There is an algorithm that,
    given $0 < \epsilon < 1/2$, integers $d, j, k \geq 1$,
    and a set $X\subset \mathbb{R}^d$ of $n$ points each having at most $j$ missing values,
    it constructs with constant probability
    an $\epsilon$-coreset for \kMeans on $X$ of size
    $m=(jk)^{O(\min\{j, k\})}\cdot (\epsilon^{-1}d\log n)^2$,
    and runs in time
    $\tilde{O}\left( (jk)^{O(\min\{ j, k \})}\cdot  nd+m \right)$.
\end{theorem}

Our coreset size is only a low-degree polynomial of $d, \epsilon$ and $\log n$,
and can thus deal with moderately-high dimension or large data set.
The dependence on $k$ (number of clusters) and $j$ (maximum number of missing values per point) is also a low-degree polynomial
as long as at least one of $k$ and $j$ is small.
Actually, we justify in Theorem~\ref{thm:lb} that this exponential dependence in $\min\{j, k\}$ cannot be further improved,
as long as the coreset size is in a similar parameter regime, i.e., the coreset size is of the form $f(j, k) \cdot \poly(\epsilon^{-1}d \log n)$.
\begin{theorem}
    \label{thm:lb}
    Consider the \kMeans with missing values problem in $\RR^d_{?}$ where each point can have at most $j$ missing coordinates. Assume there is an algorithm that constructs an $\epsilon$-coreset of size $f(j,k)\cdot \poly(\epsilon^{-1}d\log n)$, then $f(j,k)$ can not be as small as $2^{o(\min(j,k))}$.
\end{theorem}

Furthermore, the space complexity of our construction algorithm is near-linear,
and since our coreset is clearly mergeable, it is possible to apply the merge-and-reduce method~\cite{DBLP:conf/stoc/Har-PeledM04} to convert our construction into a streaming algorithm of space $\poly \log n$.
Prior to our result,
the only known coreset construction for clustering with missing values
is for the special case $j = 1$~\cite{DBLP:conf/nips/MaromF19}\footnote{
In fact,~\cite{DBLP:conf/nips/MaromF19} considers a slightly more general setting where
the input are arbitrary lines that are not necessarily axis-parallel.}
and has size $k^{O(k)} \cdot(\epsilon^{-2}d\log n)$.
Since our coreset has size $\poly(k\epsilon^{-1}d\log n)$ when $j = 1$,
it improves the dependence on $k$ over that of~\cite{DBLP:conf/nips/MaromF19}
by a factor of $k^{O(k)}$.

\paragraph{Near-linear time PTAS for \kMeans with missing values.}
Very recently, a PTAS for \kMeans with missing values,
was obtained by Eiben, Fomin, Golovach, Lochet, Panolan, and Simonov~\cite{DBLP:conf/soda/EibenFGLPS21}.
Its time bound is \emph{quadratic}, namely $O(2^{\poly(jk/\epsilon)}\cdot n^2 d)$, 
and since our coreset can be constructed in near-linear time,
we can speedup this PTAS to \emph{near-linear} time
by first constructing our coreset and then running this PTAS on the coreset.

\begin{corollary}[Near-linear time PTAS for \kMeans with missing values]
There is an algorithm that, given $0 < \epsilon < 1/2$, integers $d, j, k \geq 1$,
and a set $X\subset \mathbb{R}^d$ of $n$ points each having at most $j$ missing values,
it finds with constant probability a $(1 + \epsilon)$-approximation for \kMeans on $X$,
and runs in time
$\tilde{O}\big( (jk)^{O(\min\{ j, k \})}\cdot nd + 2^{\poly(jk/\epsilon)}\cdot d^{O(1)} \big)$.
\end{corollary}

\paragraph{Experiments.}
We implement our algorithm and validate its performance on various real and synthetic data sets in Section~\ref{sec:exp}.
Our coreset exhibits flexible tradeoffs between coreset size and accuracy,
and generally outperforms a uniform-sampling baseline and a baseline that is based on imputation,
in both error rate and stability,
especially when the coreset size is relatively small.
In particular, on each data set,
a coreset of moderate size $2000$ (which is $0.5\%$-$5\%$ of the data sets)
achieves low empirical error ($5\%$-$20\%$). 
We further demonstrate an application and use our coresets to
accelerate a Lloyd's-style heuristic adapted to the missing-values setting.
The experiments suggest that running the heuristic
on top of our coresets gives equally good solutions (error $<1\%$ relative to running on the original data set) but is much faster (speedup $>5\mathrm{x}$).

\subsection{Technical Overview}
Our coreset construction is based on the importance sampling framework
introduced by Feldman and Langberg~\cite{DBLP:conf/stoc/FeldmanL11}
and subsequently improved and generalized by~\cite{FSS20,BJKW21}.
In the framework, one first computes an importance score $\sigma_x$ for every data point $x\in X$,
and then draws independent samples with probabilities
proportional to these scores.
When no values are missing, the importance scores can be computed easily,
even for general metric spaces~\cite{DBLP:conf/fsttcs/VaradarajanX12,FSS20,BJKW21}.
However, a significant challenge with missing values
is that distances do not satisfy the triangle inequality,
hence importance scores cannot be easily computed.

We overcome this hurdle using a method introduced
by Varadarajan and Xiao~\cite{varadarajan2012near} for projective clustering
(where the triangle inequality similarly does not hold).
They reduce the importance-score computation to the construction
of a coreset for \kCenter objective;
this method is quite different from earlier approaches,
e.g.~\cite{DBLP:conf/stoc/FeldmanL11,DBLP:conf/fsttcs/VaradarajanX12,FSS20,BJKW21},
and yields a coreset for \kMeans whose size depends linearly on $\log n$
and of course on the size of the \kCenter coreset.
(Mathematically, this arises from the sum of all importance scores.)
We make use of this reduction, and thus focus on constructing (efficiently)
a small coreset for \kCenter with missing values.

An immediate difficulty is how to deal with the missing values.
We show that it is possible to find a collection of subsets of coordinates $\mathcal{I}$ (so each $I \in \mathcal{I}$ is a subset of $[d]$),
such that if we construct \kCenter coresets $S_I$ on the data set ``restricted'' to each $I \in \mathcal{I}$,
then the union of these $S_I$'s is a \kCenter coreset for the original data set with missing values.
Crucially, we ensure that each ``restricted'' data set does not contain any missing value,
so that it is possible to use a classical coreset construction for \kCenter.
Finally, we show in a technical lemma how to find a collection as necessary of size $|\mathcal{I}| \leq (jk)^{O(\min\{j, k\})}$.

Since a ``restricted'' data set does not contain any missing values,
we can use a classical \kCenter coreset construction,
and a standard construction has size $O(k\epsilon^{-d})$~\cite{DBLP:journals/algorithmica/AgarwalP02}, which is known to be tight.
We bypass this $\epsilon^{-d}$ limitation by observing that
actually $\tilde{O}(1)$-coreset for \kCenter suffices,
even though the final coreset error is $\epsilon$.
We observe that an $\tilde{O}(1)$-coreset can be constructed using a variant of
Gonzalez's algorithm~\cite{gonzalez1985clustering}.

To implement Gonzalez's algorithm,
a key step is to find the \emph{furthest} neighbor of a given subset of at most $O(k)$ points,
and a naive implementation of this runs in linear time,
which overall yields a quadratic-time coreset construction, 
because the aforementioned reduction of~\cite{varadarajan2012near}
actually requires $\Theta(n / k)$ successive runs of Gonzalez's algorithm. 
To resolve this issue, we propose a fully-dynamic implementation of Gonzalez's algorithm
so that a furthest-point query is answered in time $\poly(k\log n)$,
and the point-set is updated between successive runs instead of constructed from scratch.
Our dynamic algorithm is based on a random-projection method that was proposed
for furthest-point queries in the streaming setting~\cite{DBLP:conf/soda/Indyk03}.
Specifically, we project the (restricted) data set onto several random directions,
and on each projected (one-dimensional) data set we apply a data structure for intervals.

\subsection{Additional Related Work}

Coresets for \kMeans and \kMedian clustering have been studied extensively for two decades, and we only list a few notable results.
The first strong coresets for Euclidean \kMeans and \kMedian were given in \cite{DBLP:conf/stoc/Har-PeledM04}.
In the last decade, most work on coresets for clustering follows the importance sampling framework initiated in~\cite{langberg2010universal,DBLP:conf/stoc/FeldmanL11}.
In Euclidean space, recent work showed that coresets for \kMeans and \kMedian clustering can have size that is independent of the Euclidean dimension~\cite{FSS20,sohler2018strong,huang2020coresets}.
Beyond Euclidean space, coresets of size independent of the data-set size were constructed also for many important metric spaces~\cite{huang2018epsilon,BJKW21,cohen2021new}.
A more comprehensive overview can be found in recent surveys~\cite{phillips2017coresets,feldman2020introduction}.

Recently, attention was given also to non-traditional settings of coresets for clustering,
including 
coresets for Gaussian mixture models (GMM) \cite{lucic2017training,feldman2019coresets};
simultaneous coresets for a large family of cost functions that include both \kMedian and \kCenter \cite{braverman2019coresets};
and coresets for clustering under fairness constraints \cite{NEURIPS2019_810dfbbe}.
Also considered were settings that capture uncertainty,
for example when each point is only known to lie in a line (i.e., clustering lines) \cite{DBLP:conf/nips/MaromF19}, and when each point comes from a finite set (i.e., clustering point sets) \cite{pmlr-v119-jubran20a}.

 \section{Preliminaries}
\label{sec:prelim}
We represent a data point as a vector in $(\mathbb{R}\cup \{?\})^d$,
and a coordinate takes ``?'' if and only if it is missing.
Let $\mathbb{R}_?^d$ be a shorthand for $(\mathbb{R} \cup \{?\})^d$.
Throughout, we consider a data set $X\subset \mathbb{R}_?^d$.
The distance is evaluated only on the coordinates that are present in both $x,y$, i.e.,
$$
  \forall x, y \in \mathbb{R}_?^d,
  \qquad
  \dist(x, y) := \sqrt{\sum_{i: x_i, y_i \neq ?} {(x_i - y_i)^2}}.
$$
For $x \in \mathbb{R}_?^d$, we denote the set of coordinates that are not missing by $I_x := \{ i : x_i \neq ? \}$.
For integer $m \geq 1$, let $[m] := \{ 1, \ldots, m \}$.
For two points $p, q \in \mathbb{R}_?^d$ and an index set $I \subseteq I_p\cap I_q$,
we define the \emph{$I$-induced distance} to be $\dist_I(p, q) := \sqrt{\sum_{i\in I} (p_i-q_i)^2}$.
A point $x \in \mathbb{R}_?^d$ is called a $j$-point if it has at most $j$ missing coordinates, i.e., $|I_x| \geq d-j$.

We consider a general $k$-clustering problem called $(k, z)$-clustering,
which asks to minimize the following objective function.
This objective function (and problem) is also called
\kMedian when $z = 1$ and $\kMeans$ when $z = 2$.
\begin{definition}[\kzC]
    \label{def:kzc}
    For data set $X \subset \mathbb{R}_?^d$
    and a center set $C \subset \mathbb{R}^d$ containing $k$ (usual) points,
    let
    \begin{align*}
        \cost_z(X, C) := \sum_{x \in X}{\dist^z(x, C)}.
    \end{align*}
\end{definition}

\begin{definition}[$\epsilon$-Coreset for \kzC]
    \label{def:coreset}
    For data set $X \subset \mathbb{R}_?^d$,
    we say a weighted set $S \subseteq X$ with weight function $w : S \to \mathbb{R}_+$
    is an $\epsilon$-coreset for \kzC, if
    \begin{align*}
        \forall C \subset \mathbb{R}^d, |C| = k, \qquad
        \sum_{x \in S}{w(x) \cdot \dist^z(x, C)} \in (1 \pm \epsilon) \cdot \cost_z(X, C).
    \end{align*}
\end{definition}

\section{Coresets}

\begin{theorem}
\label{thm:main_full}
There is an algorithm that,
given as input a data set $X\subset \mathbb{R}_?^d$ of size $n=|X|$ consisting of $j$-points
and parameters $k,z\geq 1$ and $0 < \epsilon < 1/2$,
constructs with constant probability an $\epsilon$-coreset of size
$m=\tilde{O}\left(z^z \cdot \frac{(j+k)^{j+k+1}}{j^j k^{k-z-2}}\cdot \frac{(d\log n)^{\frac{z+2}{2}}}{\epsilon^2}\right)$
for \kzC of $X$, and runs in time
$\tilde{O}\left(\frac{(j+k)^{j+k+1}}{j^j k^{k-2}}\cdot nd+m\right)$.
\end{theorem}

We remark that $\frac{(j+k)^{j+k}}{j^j k^k}=(jk)^{O(\min(j,k))}$. To see this, assume $j\geq k$ w.l.o.g., so $\frac{(j+k)^{j}}{j^j}=(1+\frac{k}{j})^j\leq e^{\frac{k}{j}\cdot j}= e^k$ and $\frac{(j+k)^k}{k^k}\leq (j+k)^k$.

Theorem~\ref{thm:main_full} is the main theorem of this paper, and we present the proof in this section.
As mentioned in Section~\ref{sec:intro}, the coreset is constructed via importance sampling, by following three major steps.
\begin{enumerate}
    \item For each data point $x \in X$, compute an importance score $\sigma_x \geq 0$.
    \item Draw $N$ (to be determined later) independent samples from $X$,
    such that $x \in X$ is sampled with probability $p_x \propto \sigma_x$.
    \item Denote the sampled (multi)set as $S$, and for each $x \in S$ define its weight $w(x) := \frac{1}{p_x N}$.
    Report the weighted set $S$ as the coreset.
\end{enumerate}
The importance score $\sigma_x$ is usually defined as (an approximation) of the
\emph{sensitivity} of $x$, denoted
\begin{equation}
    \label{eqn:sensitivity_full}
    \sigma^\star_x := \sup_{C \subset \mathbb{R}^d, |C| = k}
    \frac{\dist^z(x, C)}{\cost_z(X, C)},
\end{equation}
which measures the maximum possible relative contribution of $x$ to the
objective function.

Usually, there are two main challenges with this approach.
First, the sensitivity \eqref{eqn:sensitivity_full} is not efficiently computable
because it requires to optimize over all $k$-subsets $C \subset \mathbb{R}^d$.
Second, one has to determine the number of samples $N$ (essentially the coreset size)
based on a probabilistic analysis of the event that $S$ is a coreset.
Prior work on coresets has studied these issues extensively
and developed a general framework,
and we shall use the variant stated in Theorem~\ref{thm:framework_full} below.
This framework only needs an approximation to the sensitivities $\set{\sigma^\star_x}_{x\in X}$,
more precisely it requires overestimates $\sigma_x \geq \sigma^\star_x$
whose sum $\sum_{x \in X}{\sigma_x}$ is bounded.
Moreover, it relates the number of samples $N$ to
a quantity called the \emph{weighted shattering dimension} $\sdim_{\max}$,
which roughly speaking measures the complexity of a space (set of points)
by the number of distinct ways that metric balls can intersect it.
The definition below has an extra complication of a point weight $v$,
which originates from the weight in the importance sampling procedure,
and thus we need a uniform upper bound, denoted $\sdim_{\max}$, over all possible weights.\footnote{In principle, this uniform upper bound is not necessary, and an upper bound for weights corresponding to the importance score suffices,
  but a uniform upper bound turns out to be technically easier to deal with.
}

\begin{definition}[Shattering dimension]
\label{def:sdim_full}
Given a \emph{weight} function $v : \mathbb{R}_?^d \to \mathbb{R}_+$,
let $\sdim_v(\mathbb{R}_?^d)$ be the smallest integer $t$ such that
\begin{align*}
  \forall H \subset \mathbb{R}_?^d , |H| \geq 2
  \qquad \left| \left\{ B^H_v(c, r) : c \in \mathbb{R}^d, r \geq 0 \right\} \right| \leq |H|^t,
\end{align*}
where $B^H_v(c, r) := \{ x \in H : v(x) \cdot \dist(x, c) \leq r \}$.
Let $\sdim_{\max}(\mathbb{R}_?^d) := \sup_{v : \mathbb{R}_?^d \to \mathbb{R}_+} {\sdim_v(\mathbb{R}_?^d)}$.
\end{definition}

Strictly speaking,
Theorem~\ref{thm:framework_full} has been proposed and proved only for metric spaces,
but the proof is applicable also in our setting
(where $\dist$ need not satisfy the triangle inequality),
because it only concerns the \emph{binary} relation between data points and center points (without an indirect use of a third point, e.g., by triangle inequality.)

\begin{theorem}[{\cite{FSS20}}\footnote{Our theorem statement is based on~\cite[Theorem 31]{FSS20},
  adapted to our context.
  One difference is that their theorem is about VC-dimension,
  but it is also applicable for shattering dimension.
  Another difference is that we use a more direct terminology that is specialized to metric balls in $\mathbb{R}_?^d$ instead of a general range space.
}]
\label{thm:framework_full}
Let $X \subset \mathbb{R}_?^d$ be a data set,
and let $k,z\ge 1$.
Consider the importance sampling procedure with importance scores that satisfy
$\sigma_x \geq \sigma^\star_x$ for all $x \in X$,
and with a sufficiently large number of samples
\begin{align*}
  N = \tilde{O}\left(
  \epsilon^{-2} k z^z \sdim_{\max}(\mathbb{R}_?^d) \sum_{x \in X}{\sigma_x}
  \right) .
\end{align*}
Then with constant probability it reports an $\epsilon$-coreset for \kzC.
\end{theorem}

\begin{proof}[Proof of Theorem~\ref{thm:main_full}]
Because of Theorem~\ref{thm:framework_full},
it suffices to bound $\sdim_{\max}(\mathbb{R}_?^d)$,
and to provide an efficient algorithm to estimate $\sigma_x$
whose sum is bounded.
These two components are provided in Lemma~\ref{lemma:sdim_full}
and Lemma~\ref{lemma:sensitivity_full} stated below
(their proofs appear in Sections~\ref{sec:sdim_full} and~\ref{sec:sensitivity_full}),
Plugging these two lemmas into Theorem~\ref{thm:framework_full},
the main theorem follows.
We provide an outline for the complete algorithm in Algorithm~\ref{alg:main_full}.
\begin{algorithm}[ht]
    \caption{Main algorithm}
    \label{alg:main_full}
    \begin{algorithmic}[1]
        \State run Algorithm~\ref{alg:sensitivity_full} to obtain $\sigma_x$ for $x\in X$
        \State draw $N := \tilde{O}\left(z^z \cdot \frac{(j+k)^{j+k+1}}{j^j k^{k-z-2}}\cdot \frac{(d\log n)^{\frac{z+2}{2}}}{\epsilon^2}\right)$
        independent samples $S$ from $X$, where $x \in X$ is sampled with probability $p_x \propto \sigma_x$
        \State for $x \in S$, define weight $w(x) \gets \frac{1}{p_x N}$
        \State return weighted set $S$ with weight $w$ as the coreset
    \end{algorithmic}
\end{algorithm}
\end{proof}

\begin{lemma}[Shattering dimension bound]
    \label{lemma:sdim_full}
    $\sdim_{\max}(\mathbb{R}_?^d) = O(d)$.
\end{lemma}

\begin{lemma}
\label{lemma:sensitivity_full}
There is an algorithm that, given a data set $X \subset \mathbb{R}_?^d$ of $n$ $j$-points, for \kzC
computes importance scores $\set{\sigma_x}_{x \in X}$ such that with constant probability,
\begin{itemize}
\item $\sigma_x \geq  \sigma^\star_x$ for all $x\in X$; and
\item $\sum_{x \in X}{\sigma_x} \leq O\left(\frac{(j+k)^{j+k+1}}{j^j k^{k-z-1}}\cdot \sqrt{d^z\cdot \log^{z+2} n}\right)$,
\end{itemize}
and its running time is $\tilde{O}\left(\frac{(j+k)^{j+k+2}}{j^j k^{k-2}}\cdot nd\right)$.
\end{lemma}

\subsection{Proof of Lemma~\ref{lemma:sdim_full}: Shattering Dimension of $\mathbb{R}_?^d$}
\label{sec:sdim_full}

We now prove Lemma~\ref{lemma:sdim_full},
which asserts that
$\sdim_{\max}(\mathbb{R}_?^d) = O(d)$. We remark that the shattering dimension bound for $\mathbb{R}^d$ without missing values has been proved in~\cite[Lemma 16.1]{DBLP:conf/stoc/FeldmanL11} and our proof is actually an extension of it.

\begin{proof}[Proof of Lemma~\ref{lemma:sdim_full}]
Let us verify Definition~\ref{def:sdim_full}.
Consider $H\subset \mathbb{R}_?^d$
and a weight function $v: \mathbb{R}_?^d \to \mathbb{R}_{+}$.
Recall that given $c \in \mathbb{R}^d$ and $r\geq 0$, we have
$B_v^H(c,r)=\{h\in H : v(h)\cdot \dist(h,c)\leq r\}$ and
$\dist(h,c)^2=\sum_{i\in I_h} (h_i-c_i)^2$ for $h \in H$.
We need to show that
\begin{equation}
    \label{eqn:proj_full}
    \left| \{ B_v^H(c,r) : c\in \mathbb{R}^d, r\geq 0 \} \right|
    \leq |H|^{O(d)}.
\end{equation}
Observe that
\begin{align*}
    h \in B_v^H(c, r)
    \iff v(h) \cdot \dist(h, c) \leq r
    \iff
    -r^2+\sum_{i\in I_h} (v^2(h)h_i^2+v^2(h)c_i^2-2v^2(h)h_ic_i)\leq 0.
\end{align*}
Next, we write this inequality in an alternative way, that separates
terms depending $h$ from those depending on $c$ and $r$,
more precisely as an inner-product $\langle f(h), g(c, r) \rangle \leq 0$
for vectors $f(h), g(c, r) \in \mathbb{R}^{3d+1}$.
Now consider $f : H \to \mathbb{R}^{d} \times \mathbb{R}^d \times \mathbb{R}^d \times \mathbb{R}$
and $g: \mathbb{R}^{d} \times \mathbb{R}
\to \mathbb{R}^{d} \times \mathbb{R}^d \times \mathbb{R}^d \times \mathbb{R}$
such that $f(h)=(p,q,t,-1)$, where $p,q,t\in \mathbb{R}^d$ and for $i \in [d]$
\begin{align*}
    p_i =\begin{cases}
        v^2(h)\cdot h_i^2 & \text{if } i\in I_h \\
        0 & \text{otherwise}
    \end{cases} \quad
    q_i=\begin{cases}
        v^2(h) & \text{if } i\in I_h \\
        0 & \text{otherwise}
    \end{cases} \quad
    t_i=\begin{cases}
        -2v^2(h)\cdot h_i & \text{if } i\in I_h \\
        0 & \text{otherwise}
    \end{cases}
\end{align*}
and $g(c,r)=(y,z,w,r^2)$, where $y,z,w\in \mathbb{R}^d$,
$y_i=1,z_i=c_i^2,w_i=c_i$ for $i\in [d]$.
Then we have
\begin{align*}
    h \in B_v^H(c, r) \iff \langle f(h), g(c,r) \rangle \leq 0.
\end{align*}
For a vector $t \in \mathbb{R}^{3d+1}$,
let $\proj_{-}^H(t) := \{ h \in H  : \langle f(h), t \rangle \leq 0\}$
be the subset of $H$ that has nonpositive inner-product with $t$
(it can be viewed also as projection or a halfspace).
Therefore, by (\ref{eqn:proj_full}), we have
\begin{align*}
    \left| \{ B_v^H(c,r) : c\in \mathbb{R}^d, r\geq 0 \} \right|
    = \left| \{ \proj_{-}^H(g(c, r)) : c \in \mathbb{R}^d, r\geq 0 \} \right|
    \leq \left| \{ \proj_{-}^H(t) : t \in \mathbb{R}^{3d+1} \} \right|.
\end{align*}
We observe that
\begin{align*}
    \left| \{ \proj_{-}^H(t) : t \in \mathbb{R}^{3d+1} \} \right| \leq |H|^{O(d)},
\end{align*}
since this may be related to the shattering dimension of halfspaces in $\mathbb{R}^{3d+1}$,
which is $O(d)$ and is a well-known fact in the PAC learning theory
(cf.~\cite[Chapter 7.2]{van2014probability}).
This concludes the proof of Lemma~\ref{lemma:sdim_full}.
\end{proof}

\subsection{Proof of Lemma~\ref{lemma:sensitivity_full}: Estimating Sensitivity Efficiently}
\label{sec:sensitivity_full}

We use a technique introduced by Varadarajan and Xiao~\cite{varadarajan2012near}
that reduces the sensitivity-estimation problem
to the problem of constructing a coreset for \kCenter clustering.
This coreset concept is defined as follows.
\begin{definition}
An \emph{$\alpha$-coreset for \kCenter}
of a data set $X \subset \mathbb{R}_?^d$
is a subset $Y\subseteq X$ such that
\begin{align*}
  \forall C \subset \mathbb{R}^d, |C| = k,
  \qquad
  \max_{x \in X}{\dist(x, C)} \leq \alpha \cdot \max_{y \in Y}{\dist(y, C)}.
\end{align*}
\end{definition}
Note that the error parameter $\alpha$ represents a multiplicative factor,
which is slightly different from that of $\epsilon$ in $\epsilon$-coreset for \kzC,
and roughly corresponds to $\alpha = 1 + \epsilon$.
The reasoning is that $\max_{y \in Y}{\dist(y, C)}$ for $Y \subseteq X$
is always no more than $\max_{x \in X}{\dist(x, C)}$,
and therefore we only need to measure the contraction-side error.

The reduction in Lemma~\ref{lemma:vx12_full} was presented in~\cite{varadarajan2012near}, and we restate its algorithmic steps in Algorithm~\ref{alg:sensitivity_reduct}.
This needs access to some Algorithm $\mathcal{A}$
that constructs an $\alpha$-coreset for \kCenter on a point set $X \subset \mathbb{R}^d_?$.
Each iteration $i$ calls Algorithm $\mathcal{A}$ to construct
a \kCenter coreset for the current point set $X$ (which is initially the entire data set),
assign sensitivity estimates $O(\alpha^z / i)$ to every coreset point,
and then remove these coreset points from $X$.
These iterations are repeated until $X$ is empty.

\begin{algorithm}[ht]
\caption{Sensitivity estimation from~\cite[Lemma 3.1]{varadarajan2012near} for data set $X \subset \mathbb{R}^d_?$}
\label{alg:sensitivity_reduct}
    \begin{algorithmic}[1]
        \Require algorithm $\mathcal{A}$ that constructs $\alpha$-coreset for \kCenter
        \State $i \gets 1$
        \While{$X \neq \emptyset$}
            \State $P \gets \mathcal{A}(X)$
            \For{$x \in P$}
                \State $\sigma_x \gets O(\alpha^z / i)$
            \EndFor
            \State $X \gets X \setminus P$
            \State $i \gets i + 1$
        \EndWhile
    \end{algorithmic}
\end{algorithm}

\begin{lemma}[{\cite[Lemma 3.1]{varadarajan2012near}}]
\label{lemma:vx12_full}
Suppose algorithm $\mathcal{A}$ constructs an $\alpha$-coreset of size $T = T(\alpha, d, j, k)$ for \kCenter an input $X \subset \mathbb{R}_?^d$.
Then Algorithm~\ref{alg:sensitivity_reduct} (which makes calls to this Algorithm $\mathcal{A}$)
computes sensitivities $\set{\sigma_x}$ for \kzC satisfying that
$\sigma_x \geq \sigma^\star_x$ for all $x \in X$,
and $\sum_{x \in X}{\sigma_x} \leq \alpha^z \cdot T \log|X|$.
\end{lemma}

However, there are two outstanding technical challenges.
First, there is no known construction of a small \kCenter coreset for our clustering with missing values setting.
Moreover, as can be seen from Algorithm~\ref{alg:sensitivity_reduct},
this reduction executes the \kCenter coreset construction $\frac{|X|}{T}$ times
(where $T$ is the size of the coreset as in Lemma~\ref{lemma:vx12_full}),
and when using a naive implementation of the \kCenter coreset construction,
which naturally requires $\Omega(|X|)$ time,
results overall in quadratic time, which is not very efficient.

First, to deal with question marks, we employ a certain family $\mathcal{I}$
of subset of coordinates (so each $I \in \mathcal{I}$ is a subset of $[d]$),
and we \emph{restrict} the data set $X$ on each $I \in \mathcal{I}$.
Each restricted data set (restricted on some $I$)
may be viewed as a data set in $\mathbb{R}^I$, without any question marks.
We show that the union of \kCenter coresets on all restricted data sets
with respect all to $I \in \mathcal{I}$,
forms a valid \kCenter coreset for $X$ (which has question marks),
provided that the family $\mathcal{I}$ has a certain combinatorial property.
Naturally, the size of this coreset for $X$ depends on an upper bound on $|\mathcal{I}|$.

Second, since the choice of family $\mathcal{I}$ is oblivious to the data set,
it suffices to design an efficient algorithm for \kCenter coreset for
any restricted data set.
We observe that the efficiency bottleneck in Algorithm~\ref{alg:sensitivity_reduct}
is the repeated invocation of Algorithm~$\mathcal{A}$ to construct a coreset,
even though its input changes only a little between consecutive invocations.
Hence, we design a dynamic algorithm, that maintains
a \kCenter coreset on the restricted data sets under point updates.
Our algorithm may be viewed as a variant of Gonzalez's algorithm~\cite{gonzalez1985clustering},
and we maintain it efficiently by a random projection idea that was used e.g.\ in~\cite{DBLP:conf/soda/Indyk03}.
In particular, we ``project'' the data points onto several one-dimensional lines in $\mathbb{R}^d$,
and we maintain an interval data structure (that is based on balanced trees)
to dynamically maintain the result of our variant of Gonzalez's algorithm.
We summarize the dynamic algorithm in the following lemma.

\begin{lemma}
    \label{lemma:dynamic_full}
There is a randomized dynamic algorithm with the following guarantees.
The input is a dynamic set $X \subset \RR_?^d$ of $j$-points,
such that $X$ undergoes $q$ adaptive updates (point insertions and deletions)
and the points ever added are fixed in advance (non-adaptively).
The algorithm maintains in time $\tilde{O}\left(\frac{(j+k)^{j+k+1}}{j^jk^{k}}\cdot (j+k\log q)(d+k^2\log q)\right)$ per update,
a subset $Y \subseteq X$ of size
$|Y| \leq O\left(\frac{(j+k)^{j+k+1}}{j^jk^{k-1}}\cdot \log d\right)$
such that with constant probability,
$Y$ is an $O(k\sqrt{d\log q})$-coreset for \kCenter on $X$ after every update.
\end{lemma}

The proof of the lemma can be found in Section~\ref{sec:proof_dynamic_full},
and here we proceed to the proof of Lemma~\ref{lemma:sensitivity_full}.
\begin{proof}[Proof of Lemma~\ref{lemma:sensitivity_full}]
    We plug in the dynamic algorithm in Lemma~\ref{lemma:dynamic_full} as $\mathcal{A}$
    in Lemma~\ref{lemma:vx12_full}.
    Specifically, line 3 and 7 of Algorithm~\ref{alg:sensitivity_reduct} are replaced
    by the corresponding query and update procedure.
    The detailed description can be found in Algorithm~\ref{alg:sensitivity_full}.
    \begin{algorithm}[ht]
        \caption{Efficient importance score estimation}
        \label{alg:sensitivity_full}
        \begin{algorithmic}[1]
            \State let $\mathcal{D}$ be the dynamic data structure defined in Algorithm~\ref{alg:dynamic_full}, and call $\mathcal{D}.\textsc{Init}$
            \State $\forall x \in X$, insert $x$ to $\mathcal{D}$
            \State $i \gets 1$
            \While{$X \neq \emptyset$}
                \State $P \gets \mathcal{D}.\textsc{Get-Coreset}$
                \For{$x \in P$}
                    \State $\sigma_x \gets O(\alpha^z / i)$
                \EndFor
\State $\forall x \in P$, remove $x$ from $\mathcal{D}$
                \State $i \gets i + 1$
            \EndWhile
            \State return $(\sigma_x : x\in X)$
        \end{algorithmic}
    \end{algorithm}
    
    Since $|X|=n$, and each point is inserted and deleted for exactly once, algorithm~\ref{alg:sensitivity_reduct} needs $q=O(n)$ insertions and deletions of points. Moreover, the set of points ever added is just $X$ which is fixed.
    Thus, $\alpha$ is replaced by $O(k\sqrt{d\log n})$ and $T$ is replaced by $O\left(\frac{(j+k)^{j+k+1}}{j^jk^{k-1}}\cdot \log d\right)$.
    Therefore, for \kzC, this computes $\sigma_x$ for $x\in X$
    such that $\sigma_x \geq \sigma^\star_x$, and that
    \begin{align*}
        \sum_{x \in X}{\sigma_x}
        \leq \alpha^z\cdot T\cdot \log n
        = O\left(\frac{(j+k)^{j+k+1}}{j^j k^{k-z-1}}\cdot \sqrt{ d^z\cdot \log^{z+2} n}\right).
    \end{align*}
    The total running time is bounded by $\tilde{O}\left(\frac{(j+k)^{j+k+2}}{j^jk^{k-2}}\cdot nd\right)$ for implementing $O(n)$ updates.
\end{proof}

\subsection{Proof of Lemma~\ref{lemma:dynamic_full}: Dynamic $O(1)$-Coresets for $k$-Center Clustering}
\label{sec:proof_dynamic_full}

As mentioned, the high level idea is to
identify a collection $\mathcal{I}$ of subsets of coordinates
(so each $I \in \mathcal{I}$ satisfies $I \subseteq [d]$),
construct an $\alpha$-coreset ($a$ will be determined is the later context) $Y_i$ for \kCenter on the data set
$X$ with coordinates \emph{restricted} on each $I_i \in \mathcal{I}$,
and then the union $\bigcup_i Y_i$ would be the overall $\alpha\sqrt{d}$-coreset for \kCenter on $X$.
The exact definition of restricted data set goes as follows.

\begin{definition}
For a point $p \in \mathbb{R}_?^d$ and a subset $I \subseteq I_p$,
define $p_{|I} \in \mathbb{R}^{I}$ in the obvious way,
by selecting the coordinates $\set{p_i}_{i \in I}$.
Define the \emph{$I$-restricted data set} to be
$X_{|I} := \{ p_{|I} : p \in X, I \subseteq I_p\}$.
Since each vector in $X_{|I}$ arises from a specific vector in $X$,
a subset $Y \subseteq X_{|I}$ corresponds to a specific subset of $X$,
and we shall denote this subset by $Y^{-1}$.
\end{definition}
We observe that the metric space on the restricted data set becomes a usual metric space, i.e. it satisfies the triangle inequality,
and can be realized as a point set in $\mathbb{R}^I$ which does not contain question marks.
Therefore, this reduces our goal to constructing \kCenter coresets for this usual data set.
However, the size of the coreset yielded from this approach would depend on the size of the family $\mathcal{I}$.
Hence, a key step is to identify a small set $\mathcal{I}$ such that the union of the coreset restricted on $\mathcal{I}$ is an accurate coreset.
To this end, we consider the so-called $(j, k, d)$-family of coordinates
as in Definition~\ref{def:family_full}.
This family itself is purely combinatorial, but
we will show in Lemma~\ref{lemma:family2coreset_full} that such a family
actually suffices for the accuracy of the coreset,
and we show in Lemma~\ref{lemma:family_full} the existence of a small family.

\begin{definition}
    \label{def:family_full}
A family of sets $\mathcal{I}\subset 2^{[d]}$ is called a $(j,k,d)$-family if for any $J,K\subset [d],J\cap K=\emptyset,|J|=j,|K|=k$, there exists an $I\in\mathcal{I}$ such that $I\cap J=\emptyset$ and $K\subset I$.
\end{definition}

\begin{lemma} \label{lemma:family2coreset_full}
Suppose $\mathcal{I}$ is a $(j,k,d)$-family
Let $X\subseteq \mathbb{R}_?^d$ be a set of $j$-points,
and for every $I\in \mathcal{I}$,
let $Y_I$ be an $\alpha$-coreset for \kCenter on $X_{|I}$.
Then $\cup_{I\in \mathcal{I}} Y_I^{-1}$ is an $\alpha\sqrt{d}$-coreset for $k$-Center on $X$.
\end{lemma}

\begin{proof}
It suffices to show that for any center set $C=\{c^1,\ldots,c^k\} \subseteq \mathbb{R}^d$ with $k$ points
and $x\in X$, if $\dist(x,C)\geq r$ for some $r \geq 0$, then we can find a coreset point $y\in \cup_{I\in \mathcal{I}} Y_I^{-1}$ such that $\dist(y,C)\geq \frac{r}{\alpha \sqrt{d}}$.

For $i\in [k]$, let $t_i\in \arg\max_{t\in I_x} |x_t-c_t^i|$, i.e., $t_i$ is the index of coordinate that contributes the most in distance $\dist(x,c^i)$, so $|x_{t_i}-c_{t_i}^i|\geq \frac{r}{\sqrt{d}}$. Let $K$ be any $k$-subset such that $K\subseteq I_x$ and $\{t_1,\ldots,t_k\}\subseteq K$. Since $\mathcal{I}$ is a $(j,k,d)$-family and $|I_x|\geq d-j$, by definition, there exists an $I\subseteq \mathcal{I}$ such that $K\subseteq I\subseteq I_x$.
We note that
\begin{align*}
    \dist(x_{|I},C_{|I})
    =\dist_I(x, C)
    = \min_{i \in [k]} \dist_I(x, c^i)
    \geq \min_{i \in [k]} \dist_K(x, c^i)
    \geq \min_{i \in [k]} |x_{t_i} - c_{t_i}^i|
    \geq \frac{r}{\sqrt{d}}.
\end{align*}

Since $I\subseteq I_x$, we know that $x_{|I} \in X_{|I}$. As $Y_I$ is an $\alpha$-coreset for $X_{|I}$,
we know that there exists $y \in Y_I^{-1}$ such that
$$\dist(y,C)\geq \dist_I(y, C)=\dist(y_{|I},C_{|I}) \geq\frac{\dist(x_{|I},C_{|I})}{\alpha}\geq \frac{r}{\alpha\sqrt{d}}.$$
\end{proof}

Next, we show the existence of a small $(j,k,d)$-family.
We remark that this combinatorial structure has been employed in designing fault-tolerant data structures and algorithms (cf.~\cite{dinitz2011fault,duan2021approximate,karthik2021deterministic}).
Similar bounds were obtained in their different contexts and languages, and here we provide a proof for completeness.

\begin{lemma} \label{lemma:family_full}
There is a $(j,k,d)$-family $\mathcal{I}$ of size
$O\left(\frac{(j+k)^{j+k+1}}{j^j k^k}\log d\right)$. Moreover, there is a randomized algorithm that constructs $\mathcal{I}$ in time $O(d\cdot |\mathcal{I}|)$ with probability at least $1-\frac{1}{d^{j+k}}$.
\end{lemma}

\begin{proof}
Set $t=\frac{(j+k)^{j+k+1}}{j^j k^k} \cdot 2 \log d$.
We add $t$ random sets into $\mathcal{I}$
where each random set is generated by independently including
each element of $[d]$ with probability $\frac{k}{j+k}$.
For a set $J\subseteq [d],|J|=j$ and a set $K\subseteq [d],|K|=k$ such that $J\cap K=\emptyset$,
the probability that a random set generated in the above way contains $K$ but avoids $J$, is
\begin{align*}
    \left(\frac{j}{j+k}\right)^j \cdot \left(\frac{k}{j+k}\right)^k.
\end{align*}
Since there are at most $d^{j+k}$ tuples of such $J$ and $K$,
by union bound and the choice of $t$, the probability that $\mathcal{I}$ is a $(j,k,d)$-family is at least
\begin{align*}
    1-d^{j+k} \left(1-(\frac{j}{j+k})^j \cdot (\frac{k}{j+k})^k\right)^t
    \geq 1-\frac{1}{d^{j+k}}
\end{align*}
\end{proof}

\paragraph{Gonzalez's algorithm yields \kCenter coreset for restricted data set.}
Finally, the \kCenter coreset for the restricted data set on each $I \in \mathcal{I}$ would be constructed using
an approximate version of Gonzalez's algorithm~\cite{gonzalez1985clustering}.
We note that while Gonzalez's algorithm was originally designed as an approximation algorithm
for \kCenter,
the approximate solution actually serves as a good coreset for \kCenter (see Lemma~\ref{lemma:approx_gonz_full}).
The assumption that the input forms a metric space is crucial in Lemma~\ref{lemma:approx_gonz_full},
and this is guaranteed since we run this variant of Gonzalez
only on a restricted data set which satisfies the triangle inequality.

\begin{lemma}[Approximate Gonzalez] \label{lemma:approx_gonz_full}
Let $(M,d)$ be a metric space.
Let $A\subset M$ be a set of $n$ points and consider the following variant of Gonzalez's greedy algorithm.
Set $B=\{b_0\}$ for an arbitrary $b_0\in A$.
Repeat for $k$ times, where each time we add a $c$-approximation of $B$'s furthest point into $B$.
Precisely, add $b_i\in A$ such that $c\cdot \dist(b_i,B) \geq \max_{a\in A} \dist(a,B)$ into $B$.
Then $B$ is a $(1 + 2c)$-coreset for \kCenter on $A$.
\end{lemma}

\begin{proof}
Fix a center set $C=\{c_1, \ldots, c_k\}$ with $k$ points
and let $r := \max_{b\in B} \dist(b,C)$.
Then we have $\bigcup_{i=1}^k \mathrm{Ball}(c_i,r)$ covers $B$
where $\mathrm{Ball}(x,r)=\{y:\dist(x,y)\leq r\}$ is the ball centered at $x$ with radius $r$.
It suffices to prove that $A\subseteq\bigcup_{i=1}^k \mathrm{Ball}(c_i,(2c+1)r)$.

Since $k$ balls $B(c_1,r),\cdots,B(c_k,r)$ cover $B$ and $|B|=k+1$,
by pigeonhole principle, there exists $b_i,b_j\in B,i<j$ that are contained in a same ball $B(c_i,r)$.
W.l.o.g., we assume $b_i,b_j\in B(c_1,r)$. Now fix $a\in A\setminus B$,
since $a$ has never been added into $B$, we have
\begin{align*}
    \dist(a,B)
    &\leq \dist(a, \{b_1,\ldots,b_{j-1}\})  \\
    &\leq c\cdot \dist(b_j,\{b_1,\ldots,b_{j-1}\})  \\
    &\leq c\cdot \dist(b_i,b_j)  \\
    &\leq c\cdot (\dist(b_i,c_1)+\dist(b_j,c_1))  \\
    &\leq 2cr.
\end{align*}

Thus $A\subseteq \bigcup_{i=1}^{k+1} \mathrm{Ball}(b_i,2cr)\subseteq \bigcup_{i=1}^k \mathrm{Ball}(c_i,(2c+1)r)$.
\end{proof}

\paragraph{Dynamic implementation of Gonzalez's algorithm.}
To make this \kCenter coreset construction dynamic,
we adapt the random projection
technique to Gonzalez's algorithm, so that it suffices to dynamically execute
Gonzalez's algorithm on a set of one-dimensional lines in $\mathbb{R}^d$.

\paragraph{Random projection.}
We call a sample from the $d$-dimensional standard normal distribution $N(0,I_d)$ a $d$-dimensional \emph{random vector} for simplicity.
To implement (the variant of) Gonzalez's algorithm as in Lemma~\ref{lemma:approx_gonz_full}
in the dynamic setting,
we project the point set to several random vectors
and use one dimensional data structure to construct \kCenter coreset in
each of the one dimensional projected data set.

Note that the key step in Gonzalez's algorithm is the furthest neighbor search,
and we would show that our projection method eventually yields
an $O(k\sqrt{\log n})$-approximation of the furthest neighbor with high probability.
The following two facts about normal distribution are crucial in our argument, and Lemma~\ref{lemma:projection_full} is our main technical lemma.
\begin{fact}
    Let $u\in \mathbb{R}^d$ and let $v\sim N(0,I_d)$ be a random vector,
    then $\langle u, v/|u| \rangle \sim N(0,1)$.
\end{fact}
\begin{fact} \label{Gaussian_full}
Let $Z\sim N(0,1)$, then there exists some universal constant $c>0$ such that $P[|Z|\leq \frac{1}{k}]\leq \frac{c}{k}$, and $P[|Z|\geq t]\leq e^{-c\cdot t^2}$ for any $t>0$.
\end{fact}

\begin{lemma} \label{lemma:projection_full}
Let $X\subset \mathbb{R}^d,|X|=n$, $\delta > 0$ and integer $k \geq 1$.
Let $\mathcal{V}$ be a collection of $t=O(k\log n+\log\delta^{-1})$  random vectors in $\mathbb{R}^d$. Then with probability $1-\delta$, for every $C\subseteq X,|C|\leq k$ and every $x\in X$, there exists a vector $v\in \mathcal{V}$ such that (i)
$|x\cdot v-c\cdot v|\geq \Omega(\frac{1}{k})\cdot  \|c-x\|_2$ for every $c\in C$ and (ii) $|a\cdot v-b\cdot v|\leq O(\sqrt{\log n})\cdot \|a-b\|_2$ for every $a,b\in X$.
\end{lemma}

\begin{proof}
Fix a subset $C\subseteq X,|C|\leq k$, a point $x$ and a random vector $v$. For every $c\in C$, since $(c-x)\cdot v/\|c-x\|_2\sim N(0,1)$, by Fact~\ref{Gaussian_full}, the probability that $|c\cdot v-x\cdot v|\geq \Omega(\frac{1}{k})\cdot \|c-x\|_2
$ is at least $1-\frac{1}{4k}$. For every $a,b\in X$, since $(a-b)\cdot v/\|a-b\|_2\sim N(0,1)$, by Fact~\ref{Gaussian_full}, the probability that $|a\cdot v-b\cdot v|\leq \sqrt{\log n}\|a-b\|_2$ is at most $\frac{1}{4n^2}$.

Since there are $k$ choices of $c\in C$ and at most $n^2$ choices of $a,b\in X$, by union bound, with probability at least $1-k\cdot \frac{1}{4k}-n^2\cdot \frac{1}{4n^2}=\frac{1}{2}$, the following two events hold,
(i) $|x\cdot v-c\cdot v|\geq \Omega(\frac{1}{k})\cdot  \|c-x\|_2$ for every $c\in C$ and (ii) $|a\cdot v-b\cdot v|\leq O(\sqrt{\log n})\cdot \|a-b\|_2$ for every $a,b\in X$.

Now since $\mathcal{V}$ contains $t$ random vectors, the probability that there exists one vector $v\in \mathcal{V}$ that satisfies (i) and (ii) is at least
$
1-\frac{1}{2^t}
$.

Finally, by union bound, since there are at most $(n+1)^{k+1}$ choices of $C\subseteq X,|C|=k$ and $x\in X$, the probability such that for every $C$ and $x$, there exists $v\in \mathcal{V}$ such that (i) and (ii) happen is at least $1-\frac{(n+1)^{k+1}}{2^t}\geq 1-\delta$.
\end{proof}

In the next lemma, we present a dynamic algorithm that
combines the random projection idea with a one-dimensional data structure.
This combining with the $(j, k, d)$-family idea would immediately imply Lemma~\ref{lemma:dynamic_full}.

\begin{lemma}
    \label{lemma:dynamicI_full}
    There is a dynamic algorithm that for every $P \subseteq \mathbb{R}^m$ subject to at most $q$ adaptive point insertions and deletions where the set of points ever added is fixed in advance,
    and every $\delta > 0$,
    maintains set $Q \subseteq P$ with $|Q| \leq k+1$ such that with probability at least $1 - \delta$,
    $Q$ is an $O(k\sqrt{\log q})$-coreset for \kCenter on $P$ after every update,
    in time $O\big((k^2\log q+m)(k\log q+\log\delta^{-1})\big)$ per update.
\end{lemma}
\begin{proof}[Proof of Lemma~\ref{lemma:dynamic_full}]
We present our dynamic algorithm in Algorithm~\ref{alg:dynamic_full}.
    \begin{algorithm}[ht]
        \caption{Dynamic \kCenter coreset with missing values}
        \label{alg:dynamic_full}
        \begin{algorithmic}[1]
            \Procedure{Init}{}
                \State let $\mathcal{I}$ be a $(j, k, d)$-family generated
                by sampling, as in Lemma~\ref{lemma:family_full}
                
                \Comment{$|\mathcal{I}| = O\left(\frac{(j+k)^{j+k+1}}{j^j k^k}\log d\right)$}
                \State $\forall I \in \mathcal{I}$, initialize data structure $\mathcal{D}_I$ using 
                    Algorithm~\ref{alg:dynamicI_full} (Lemma~\ref{lemma:dynamicI_full}) with failure probability $\delta := \Theta\left(\frac{1}{|\mathcal{I}|}\right)$,
                    and initialize $Y_I = \emptyset$
            \EndProcedure
            \Procedure{Update}{$x \in \mathbb{R}^d_{?}$}
                \For{$I \in \mathcal{I}$}
                    \State $\mathcal{D}_I.\textsc{Update}(x_{|I})$
                    \State $Y_I \gets \mathcal{D}_I.\textsc{Get-Coreset}(k)$
                \EndFor
                \Comment{we use $\textsc{Update}$ and $\text{Get-Coreset}$ in Algorithm~\ref{alg:dynamicI_full}}
            \EndProcedure
            \Procedure{Get-Coreset}{}
                \State return $\bigcup_{I \in \mathcal{I}}{Y_I^{-1}}$
                \Comment{as in Lemma~\ref{lemma:family2coreset_full}}
            \EndProcedure
        \end{algorithmic}
    \end{algorithm}

    \paragraph{Analysis.}
    Since we pick $\delta = \Theta\left(\frac{1}{|\mathcal{I}|}\right)$ for all $\mathcal{D}_I$'s,
    with constant probability all data structures $\mathcal{D}_I$'s succeed simultaneously.
    The running time follows immediately from Lemma~\ref{lemma:family_full} and
    Lemma~\ref{lemma:dynamicI_full}.
    The coreset accuracy follows from Lemma~\ref{lemma:family2coreset_full} and
    Lemma~\ref{lemma:dynamicI_full} (noting that we need to suffer a $\sqrt{d}$ factor because of Lemma~\ref{lemma:family2coreset_full}).
\end{proof}
\begin{proof}[Proof of Lemma~\ref{lemma:dynamicI_full}]
We assume there is a data structure $\mathcal{T}$ that maintains a set of real numbers
    and supports the following operations, all running in $O(\log n)$ time
    where $n$ is the number of elements currently present in the structure.
    \begin{itemize}
        \item $\textsc{Remove}(x)$: Remove an element $x$ from the structure.
        \item $\textsc{Add}(x)$: Add an element $x$ to the structure.
        \item $\textsc{UpperBound}(x)$: Return the largest element that is at most $x$.
        \item $\textsc{LowerBound}(x)$: Return the smallest element that is at least $x$.
    \end{itemize}
    Note that such $\mathcal{T}$ may be implemented by using a standard balanced binary tree.
    \paragraph{Furthest point query.}
    We also need $\textsc{Furthest}(C)$ query,
    where $C \subset \mathbb{R}$ and it asks for an element $x$ that has the largest distance to $C$ (and it should return an arbitrary element if $C = \emptyset$).
    This $\textsc{Furthest}(C)$ can be implemented by
    using $O(|C|)$ many $\textsc{UpperBound}$ and $\textsc{LowerBound}$ operations,
    which then takes $O(|C| \log n)$ time in total.
    To see this, assume $C=\{c_1,\ldots,c_k\}$ where $c_1\leq \ldots \leq c_k$ then
    the clusters partitoned by $C$ is $(-\infty,\frac{1}{2}(c_1+c_2)],(\frac{1}{2}(c_1+c_2),\frac{1}{2}(c_2+c_3)],\cdots,(\frac{1}{2}(c_{k+1}+c_k),+\infty)$ and we can find the potential furthest points in each cluster by querying the following,
    \begin{eqnarray*}
    &&\textsc{UpperBound}(-\infty),\textsc{LowerBound}\left(\frac{1}{2}(c_1+c_2)\right),\\
    &&\textsc{UpperBound}\left(\frac{1}{2}(c_1+c_2)\big),\textsc{LowerBound}\big(\frac{1}{2}(c_2+c_3)\right)\\
    &&\ldots\\
    &&\textsc{UpperBound}\left(\frac{1}{2}(c_{k+1}+c_k)\right),\textsc{LowerBound}(+\infty)
    \end{eqnarray*}
    and the furthest point to $C$ among the above $2k=O(|C|)$ many points is what we seek for.

    The dynamic algorithm is presented in Algorithm~\ref{alg:dynamicI_full}.
    The algorithm samples a set of independent random vectors $\mathcal{V}$ (in a data oblivious way),
    then creates an above-mentioned interval structure $\mathcal{T}_v$
    for each $v \in \mathcal{V}$.
    When we insert/delete a point $x$,
    the update is performed on every $\mathcal{T}_v$ with the projection $\langle x, v \rangle$.
    The coreset for the current data set $P$ can be computed on the fly
    by simulating the Gonzalez's algorithm.
    In particular, this is where the Furthest query is used,
    and we find an approximate furthest point in $P$ by taking the furthest
    point in each $\mathcal{T}_v$, and select the one that is the relative furthest in $P$.

    \begin{algorithm}[ht]
        \caption{Dynamic Gonzalez's algorithm}
        \label{alg:dynamicI_full}
        \begin{algorithmic}[1]
            \Procedure{Init}{} \Comment{initialize an empty structure}
                \State $l \gets O(k\log q+\log \delta^{-1})$,
                and draw $l$ independent random vectors in $\mathbb{R}^m$, denotes as $\mathcal{V}$
                \State initialize $\mathcal{T}_v$ for each $v \in \mathcal{V}$
\EndProcedure
            \Procedure{Update}{$x$}
                \State insert/delete $\langle x, v\rangle$
                for each $v \in \mathcal{V}$
            \EndProcedure
            \Procedure{Get-Coreset}{$k$}
                \State $Q \gets \emptyset$
                \For{$i = 1, \ldots, k+1$}
                    \State for $v\in \mathcal{V}$, let $x_v \in P$ satisfy
                    $\langle x_v, v \rangle = \mathcal{T}_v.\textsc{Furthest}(\langle Q, v\rangle)$

                    \Comment{where $\langle Q, v\rangle := \{ \langle x, v \rangle : x \in Q \}$}
                    \State $v^\star \gets \arg\max_{v \in \mathcal{V}} \dist(x_v, Q)$
                    \State $Q \gets Q \cup \{ x_{v^\star} \}$
                \EndFor
                \State \Return $Q$
            \EndProcedure
        \end{algorithmic}
    \end{algorithm}

    \paragraph{Analysis.}
    Let $A$ be the set of points ever added, so $|A|\leq q$. Recall that $A$ is fixed in advance.
    By applying Lemma~\ref{lemma:projection_full} in $A$, we know that
    with probability $1-\delta$, the following event $\mathcal{E}$ happens.
    For every $C\subseteq A, |C|\leq k$, every $x\in A$,
    there exists $v\in \mathcal{V}$, such that
    \begin{itemize}
        \item[(i)] $|\langle c - x, v \rangle|\geq \Omega(\frac{1}{k})\cdot \|x-c\|_2$ for every $c\in C$, and
        \item[(ii)] $|\langle a - b, v \rangle|\leq O(\sqrt{\log q})\cdot \|a-b\|_2$ for every $a,b\in A$.
    \end{itemize}

    Now condition on $\mathcal{E}$.
    Suppose the current point set is $P$.
    Suppose we run the \textsc{Get-Coreset} subroutine and we query $\mathcal{T}_v.\textsc{Furthest}(\langle Q, v \rangle)$ for some $v$ and $Q$.
    Suppose $x\in P\subseteq A$ is the current furthest point to $Q$. Because of $\mathcal{E}$,
    there exists a vector $v\in \mathcal{V}$ such that (i) and (ii) hold.
    By (i), we have that $\dist(\langle x, v \rangle, \langle Q, v\rangle)\geq \Omega(\frac{1}{k})\cdot \dist(x,Q)$.
    By (ii), we know that for any $p\in P$ and $c\in Q$, $|\langle p - c, v \rangle|\leq O(\sqrt{\log q})\|p-c\|_2$,
    so $\dist(\langle p, v \rangle, \langle Q, v\rangle)\leq O(\sqrt{\log q})\cdot \dist(p, Q)$.
    So if $\mathcal{T}_v.\textsc{Furthest}(\langle Q, v\rangle)$ returns an answer $\langle p, v \rangle$, we know that
    \begin{align*}
        \dist(p,Q)
        \geq \frac{\dist(\langle p, v \rangle, \langle Q, v\rangle)}{O(\sqrt{\log q})}
        \geq\frac{\dist(\langle x, v \rangle, \langle Q, v\rangle)}{O(\sqrt{\log q})}
        \geq \Omega\left(\frac{1}{k\sqrt{\log q}}\right)\cdot \dist(x,Q).
    \end{align*}
    Thus, $p$ is an $O(k\sqrt{\log q})$-approximation of the furthest point to $Q$.
    This combining with Lemma~\ref{lemma:approx_gonz_full}.
    implies the error bound.

    \paragraph{Running time.}
    For the running time, we note that for each update of $P$, we need to update $\mathcal{T}_v$ for each $v\in \mathcal{V}$ accordingly.
    Thus we need to pay $O(lm)$ time (recalling that
    $l = O(k\log q+\log \delta^{-1})$ was defined in Algorithm~\ref{alg:dynamicI_full})
    to compute all the inner products and $O(l \log q)$ time to update all $\mathcal{T}_v$'s.
    The main loop in \textsc{Get-Coreset} requires $O(kl)$ many
    $\textsc{Furthest}(\cdot)$ queries and this runs in $O(k^2l\log q)$ time in total.
    In conclusion, the running time of each update (and maintaining coreset) is bounded by
    \begin{align*}
        O\left((k^2\log q+m)\cdot l\big)=O\big((k^2\log q+m)(k\log q+\log\delta^{-1})\right).
    \end{align*}
\end{proof}

 \section{Experiments}
\label{sec:exp}
We implement our proposed coreset construction algorithm, and we evaluate its performance on real and synthetic datasets.
We focus on \kMeans with missing values, and we examine the speedup for a Lloyd's-style heuristic.
In addition to measuring the absolute performance of our coreset,
we also compare it with a) uniform sampling baseline,
which is a naive way to construct coresets, and
b) an imputation-based baseline where missing values are filled
in by random values and then a standard importance-sampling coreset construction (cf.~\cite{DBLP:conf/stoc/FeldmanL11})
is run on top of it.
We implement the algorithms using C++ 11, on a laptop with Intel i5-8350U CPU and 8GB RAM.

\paragraph{Datasets.}
We run our experiments on three real datasets and one synthetic dataset.
Below, we briefly describe how we process and choose the attributes of the dataset,
and the parameters of the datasets after processing are summarized in Table~\ref{table:dataset}.
\begin{compactenum}
\item Russian housing~\cite{RussianHouse} is a dataset on Russian house market. We pick four main numerical attributes of the houses which are the full area, the live area, the kitchen area and the price,
    and the price attribute is divided by $10^5$ so as it lies in the similar range of other attributes. Three columns regarding area contain missing values, and the price column doesn't contain any missing value.
\item KDDCup 2009~\cite{KDD09} is a dataset on customer relationship prediction. We pick $31$ numerical attributes that have similar magnitudes. Each column contains missing values.
    \item Vertical farming~\cite{VerticalFarming} is a dataset about cubes which are used for advanced vertical farming. We include all of four numerical attributes of the dataset. Each column contains missing values.
    \item Synthetic dataset. We generate a large synthetic dataset to validate our algorithm's scalability.
    Data points are randomly generated so that $97\%$ of them are in a square and $3\%$ of them are far away from the square.
    After that, we delete $25\%$ of attributes at random.
    We remark that the $3\%$ far away points is to make the dataset less uniform which prevents it from being trivial for clustering.
\end{compactenum}

\begin{table}[t]
  \caption{Parameters of the datasets. $n$ is the number of data points,
  $d$ is the dimension, $k$ is the number of clusters, $j$ is the maximum number of missing coordinates for each point. $n$,$d$,$j$ are given, and $k$ is chosen by us.}
  \label{table:dataset}
  \centering
  \begin{tabular}{lllll}
    \toprule
    Data set     & $n$ & $d$ & $k$ & $j$   \\
    \midrule
    Russian housing & 30471 & 4 & 3 & 3       \\
    KDD cup     & 50000 & 31 & 5 & 30     \\
    Vertical farming & 400180 & 4 & 2 & 4 \\
    Synthetic & 200000 & 3& 3 & 3 \\
    \bottomrule
  \end{tabular}
\end{table}

\paragraph{Implementation notes.}
In our experiments, we follow a standard practice of fixing coreset size in each experiment (cf.~\cite{baker2020coresets,pmlr-v119-jubran20a}). Recall that when computing the importance score,
our algorithm chooses a family $\mathcal{I} $ of subsets of coordinates and
work on each restricted data set $X_{|I}$ for $I \in \mathcal{I}$. For a fixed size coreset, the family size $|\mathcal{I}|$ is a parameter that needs to be optimized. In Figure~\ref{fig:tune}, we plot the empirical error (defined in (\ref{eqn:empirical_error}), Section~\ref{sec:accuracy}) for the Russian housing dataset
with respect to the family size $|\mathcal{I}|$. Although Lemma 3.4 gives a theoretical upper bound on $|\mathcal{I}|$ but our experiments suggest that a much smaller size $|\mathcal{I}|=20$ is optimal in this case.

\begin{figure}[t]
    \centering
    \captionsetup{font=small}
    \includegraphics[width=0.4\textwidth]{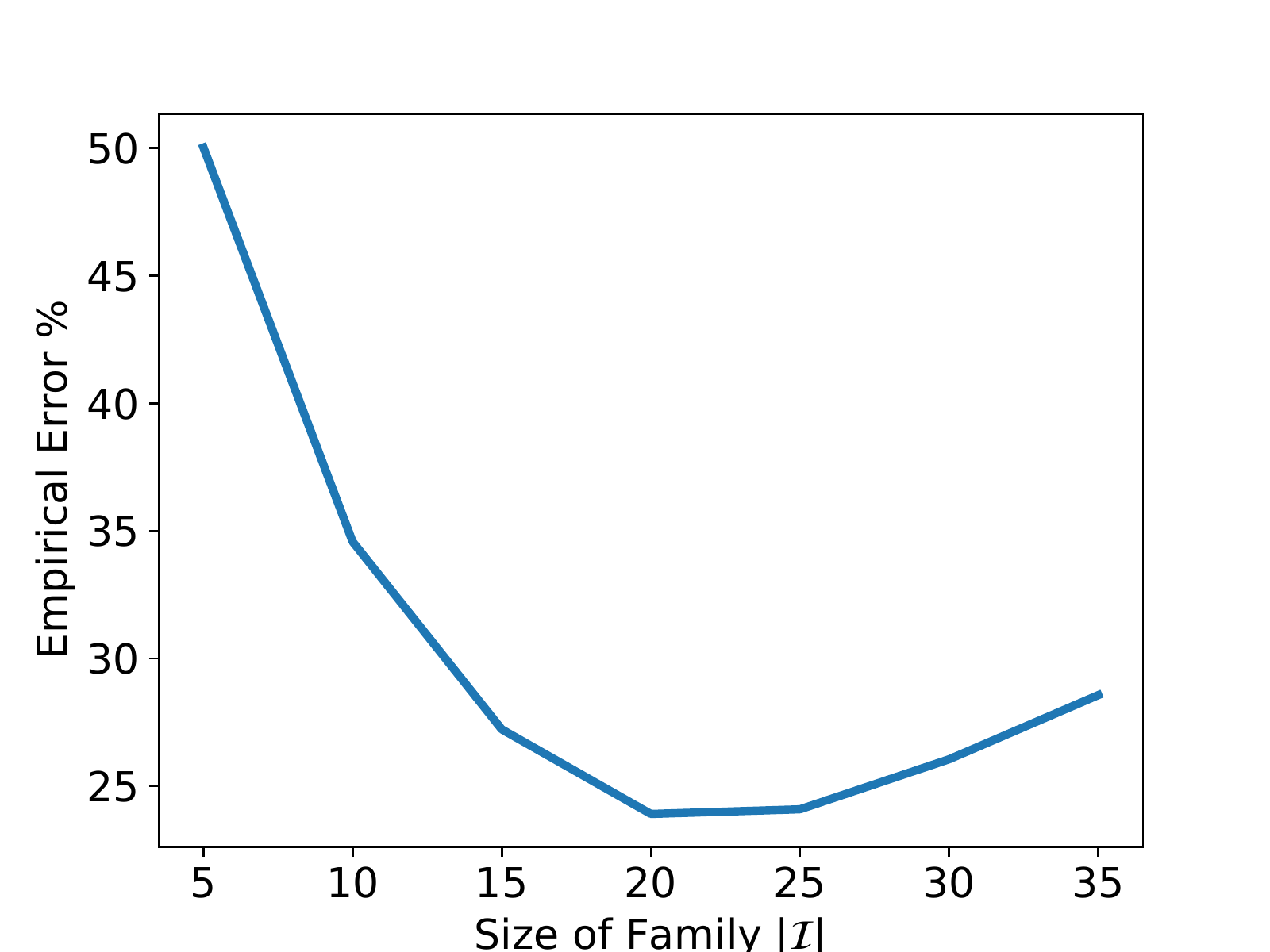}
    \caption{The average empirical error of Russian housing data set with respect to family size $|\mathcal{I}|$ on $10$ independent experiments.}
    \label{fig:tune}
\end{figure}

\subsection{Accuracy of Coresets}
\label{sec:accuracy}

We evaluate the accuracy versus size tradeoff of our coresets.
Since the coreset should preserve the clustering cost for \emph{all} centers,
we evaluate the accuracy by testing the \emph{empirical error} on a selected set of centers $\mathcal{C}$.
Namely, for a data set $X$, a coreset $D\subseteq X$ and a collection of center sets $\mathcal{C}$,
we define the empirical error of $D$ as
\begin{align}
    \mathrm{err}(D)=\max_{C\in \mathcal{C}}\frac{|\mathrm{cost}(D,C)-\mathrm{cost}(X,C)|}{\mathrm{cost}(X,C)}.
    \label{eqn:empirical_error}
\end{align}
We use a randomly selected collection of centers $\mathcal{C}$ that consists of $100$ randomly generated $k$-subset $C \subset \mathbb{R}^d$.
Since both the evaluation method and the algorithm has randomness,
we run the experiment for $T=10^3$ times with independent random bits and report the average empirical error to make it stable.
We choose $20$ different coreset sizes from $200$ to $9700$ in a step size of $500$,
and report the corresponding average empirical error.

\paragraph{Results.}
We report the size versus accuracy tradeoff of our coreset for all four datasets in Figure~\ref{fig:accu},
and record the standard deviation in Figure~\ref{fig:std}.
We compare these results against the abovementioned
uniform sampling and imputation baseline.
As can be seen from the figures,
the accuracy of our coreset improves when the size increases,
and we achieve $5\%$-$20\%$ error using only $2000$ coreset points
(which is within $0.5\% - 5\%$ of the datasets).
This $5\%$-$20\%$ error is likely to be enough for practical use,
since practical algorithms for \kMeans are approximation algorithms anyway.
Our coresets generally outperform both the uniform sampling and imputation baselines
on almost every coreset sample size, and the advantage is more significant when the coreset size is relatively small.
Moreover, our coresets have a much lower variance.

\begin{figure}
    \centering
    \captionsetup{font=small}
    \begin{subfigure}[b]{0.4\textwidth}
        \centering
        \includegraphics[width=\textwidth]{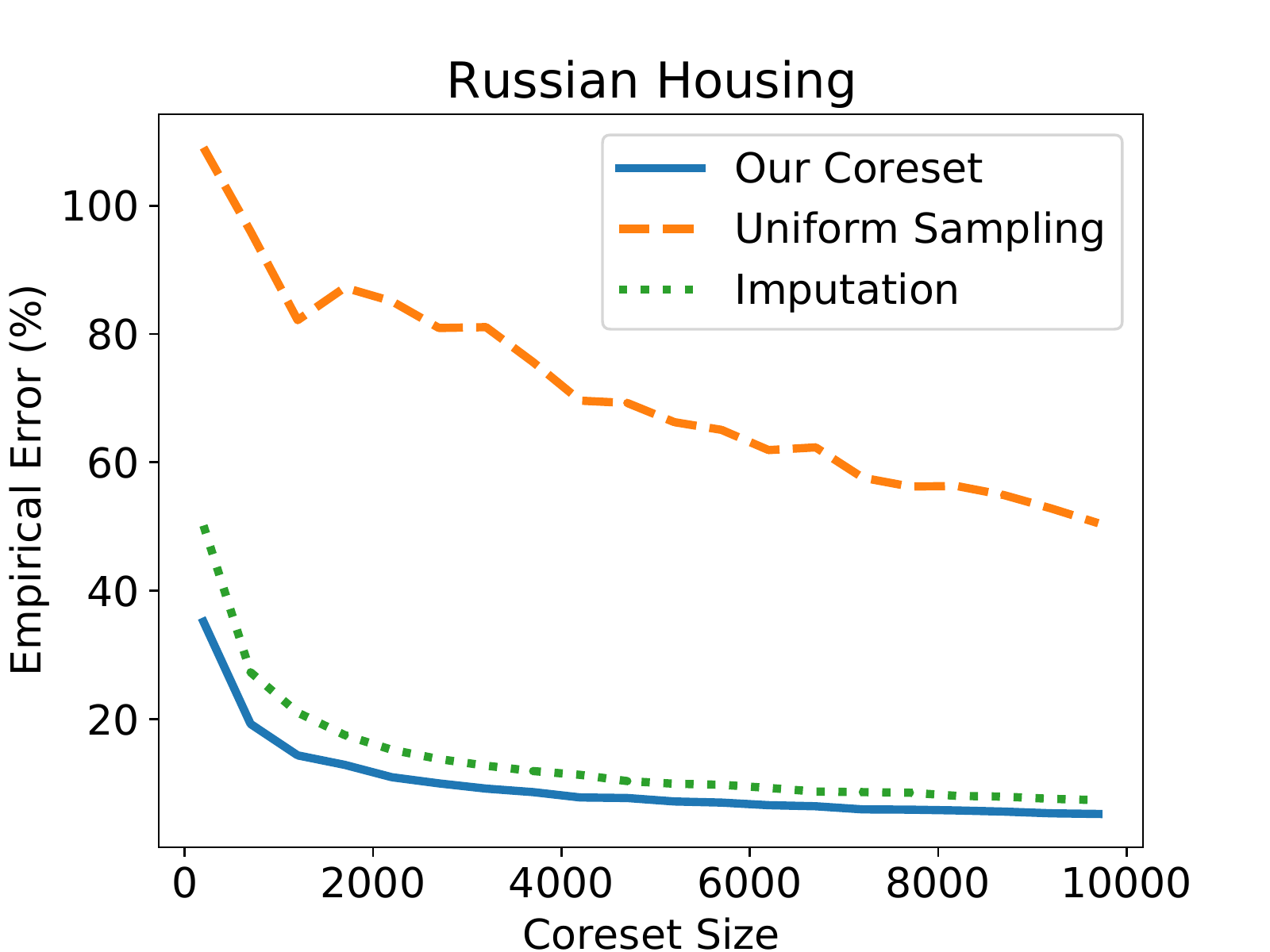}
        \caption{}
        \label{fig:rh_accu}
    \end{subfigure}
    \begin{subfigure}[b]{0.4\textwidth}
        \centering
        \includegraphics[width=\textwidth]{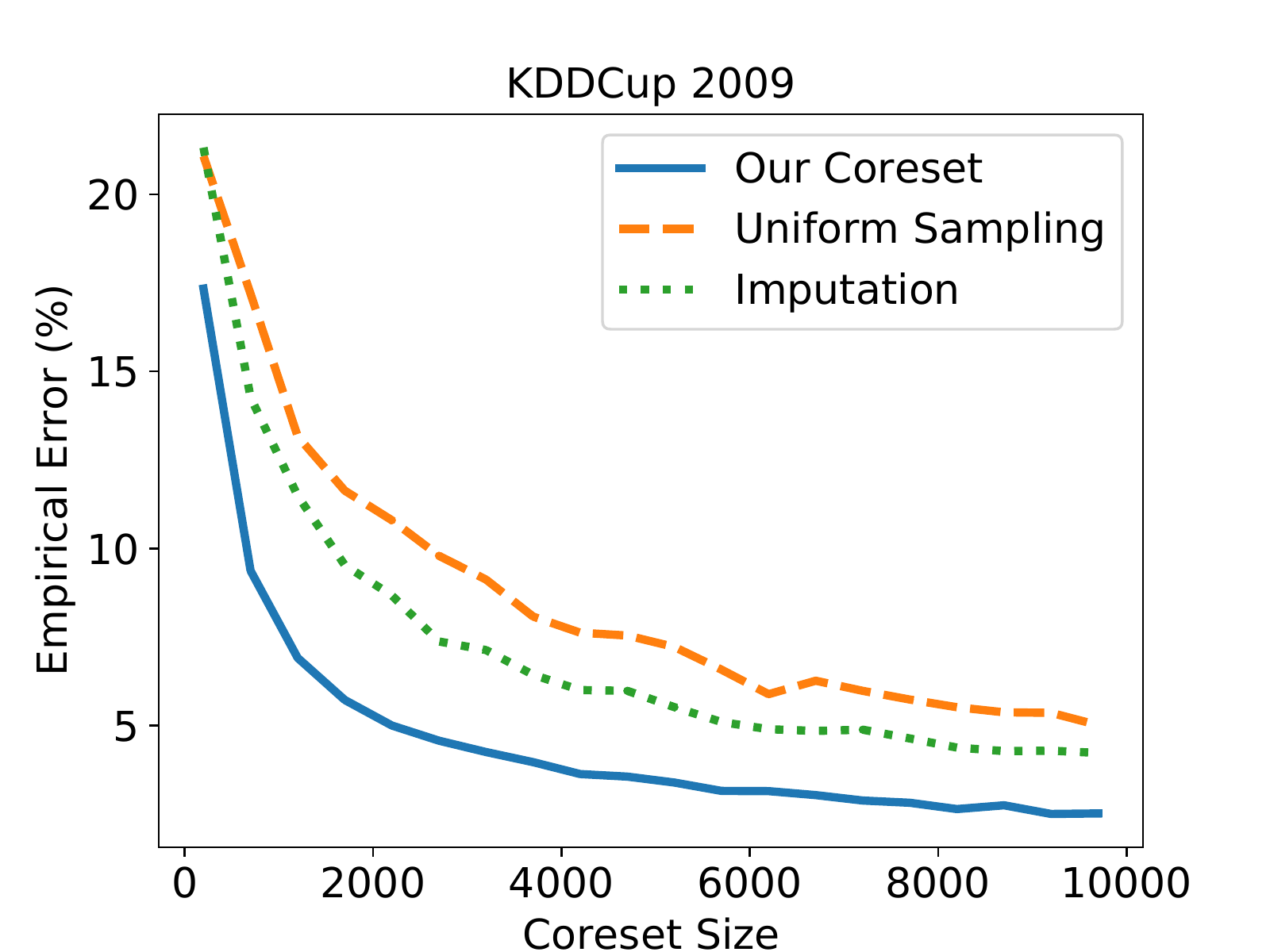}
        \caption{}
    \end{subfigure}
    \begin{subfigure}[b]{0.4\textwidth}
        \centering
        \includegraphics[width=\textwidth]{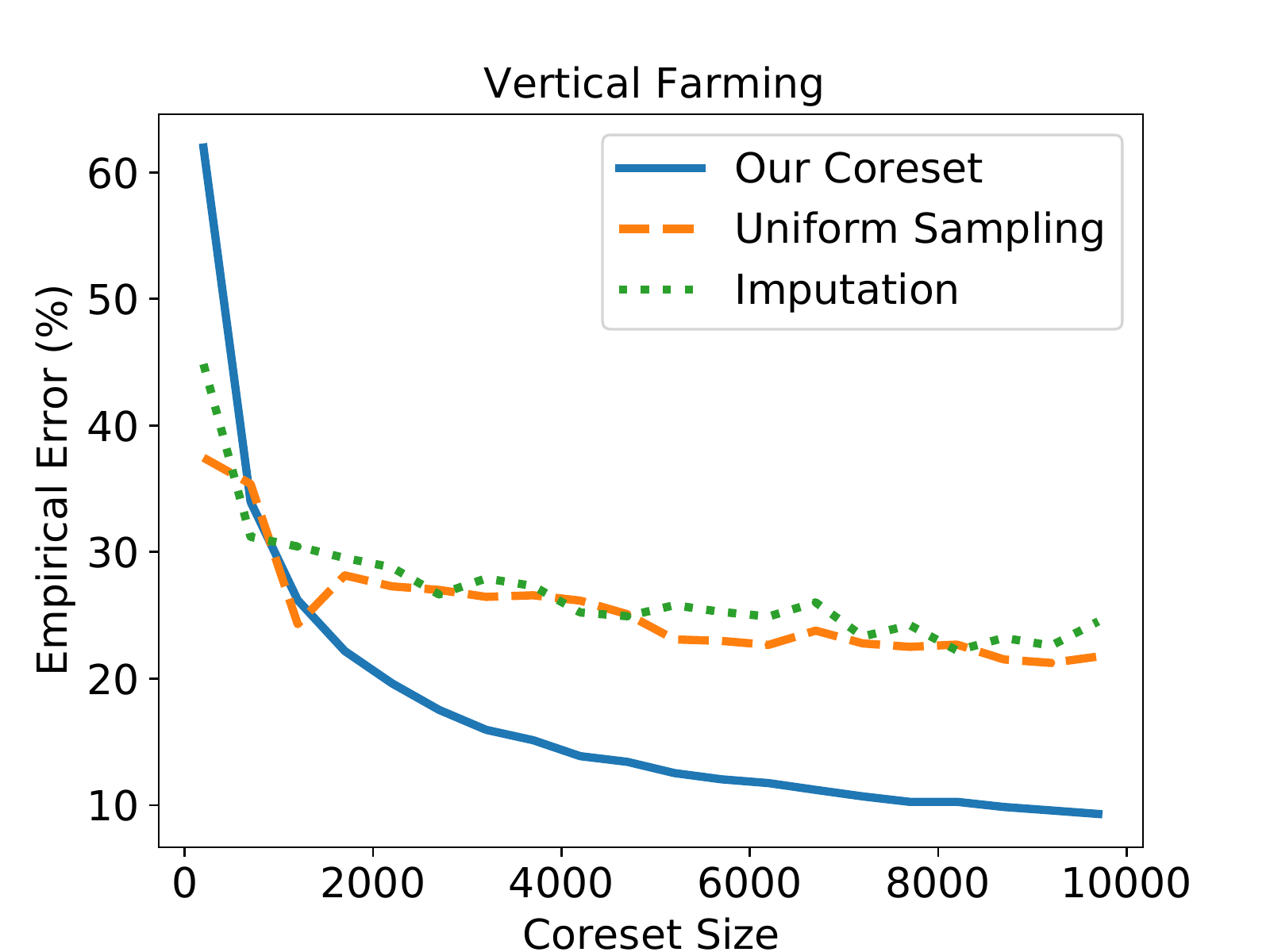}
        \caption{}
    \end{subfigure}
    \begin{subfigure}[b]{0.4\textwidth}
        \centering
        \includegraphics[width=\textwidth]{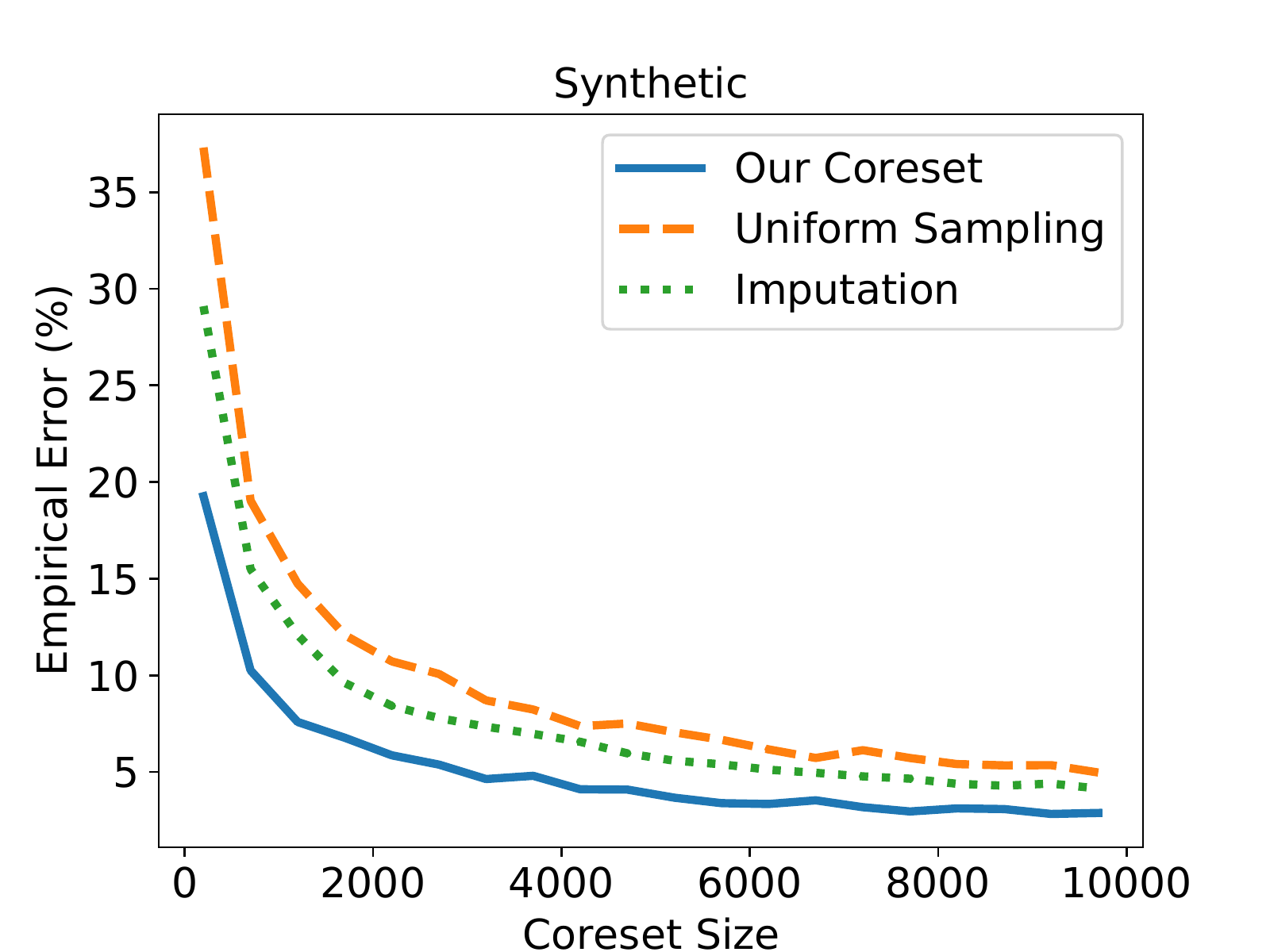}
        \caption{}
    \end{subfigure}
    \caption{Accuracy evaluation for the datasets with respect to varing coreset sizes, compared against uniform sampling and imputation baselines.}
    \label{fig:accu}
\end{figure}

\begin{figure}
    \centering
    \captionsetup{font=small}
    \begin{subfigure}[b]{0.4\textwidth}
        \centering
        \includegraphics[width=\textwidth]{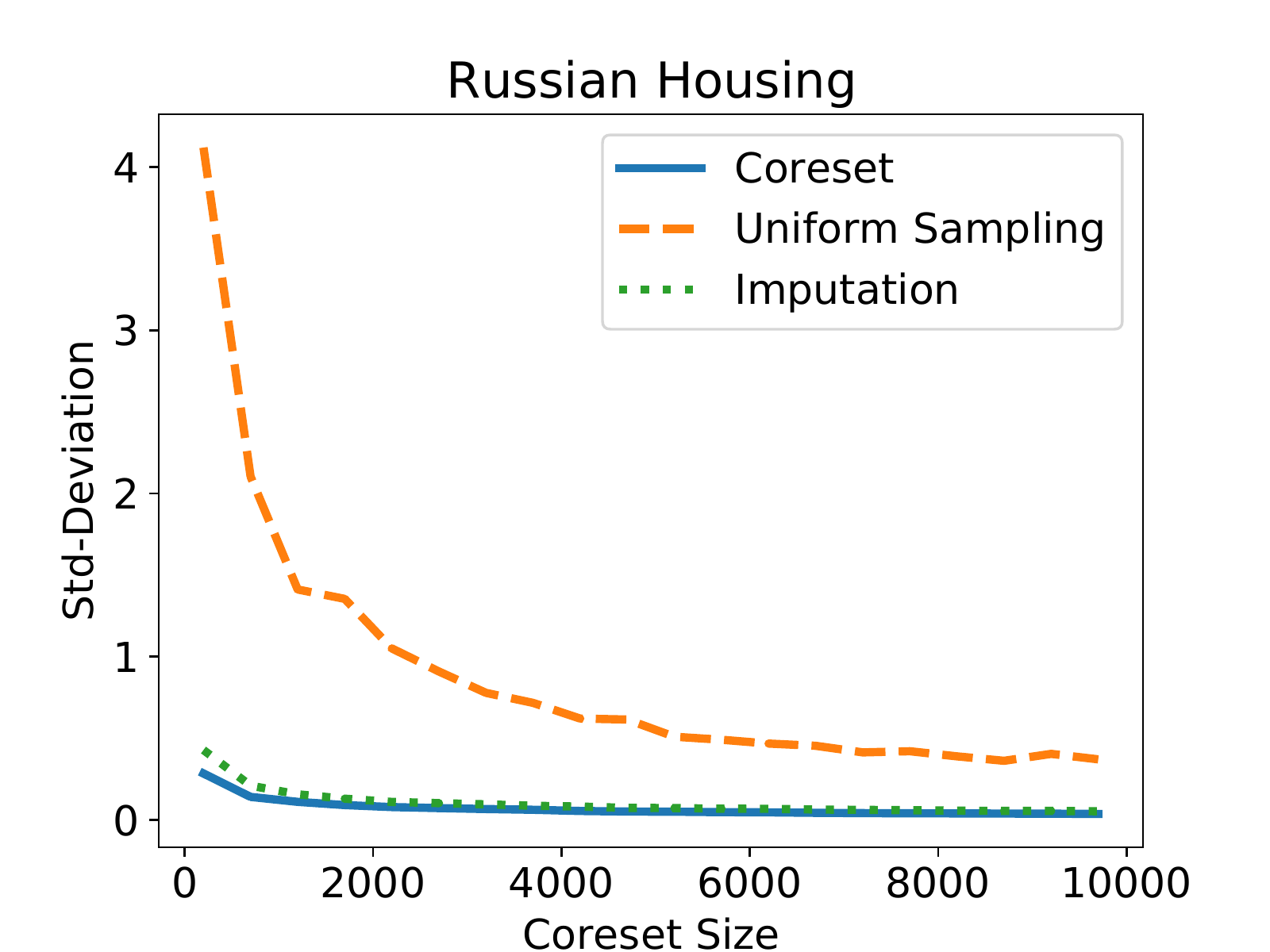}
        \caption{}
    \end{subfigure}
    \begin{subfigure}[b]{0.4\textwidth}
        \centering
        \includegraphics[width=\textwidth]{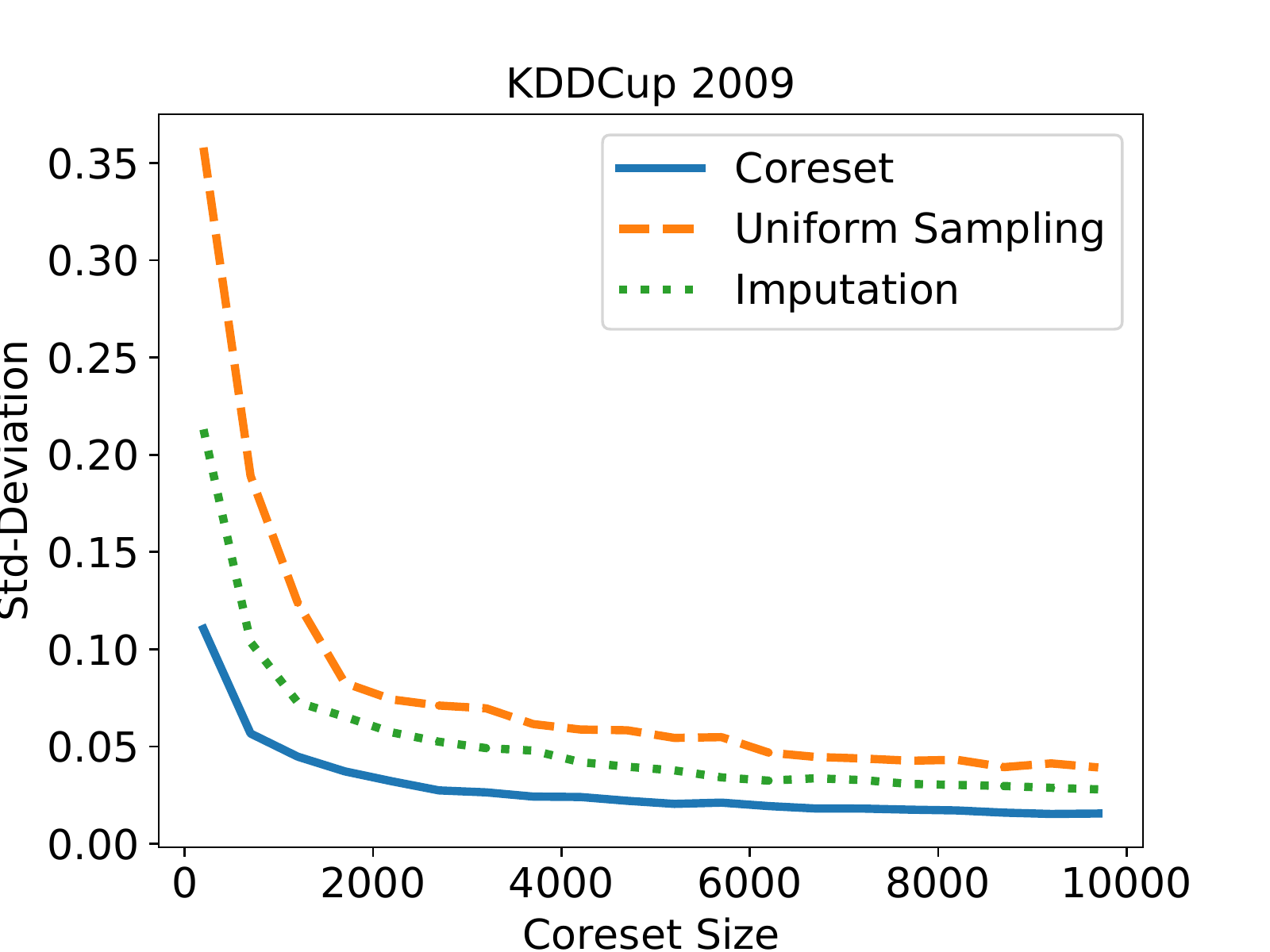}
        \caption{}
    \end{subfigure}
    \begin{subfigure}[b]{0.4\textwidth}
        \centering
        \includegraphics[width=\textwidth]{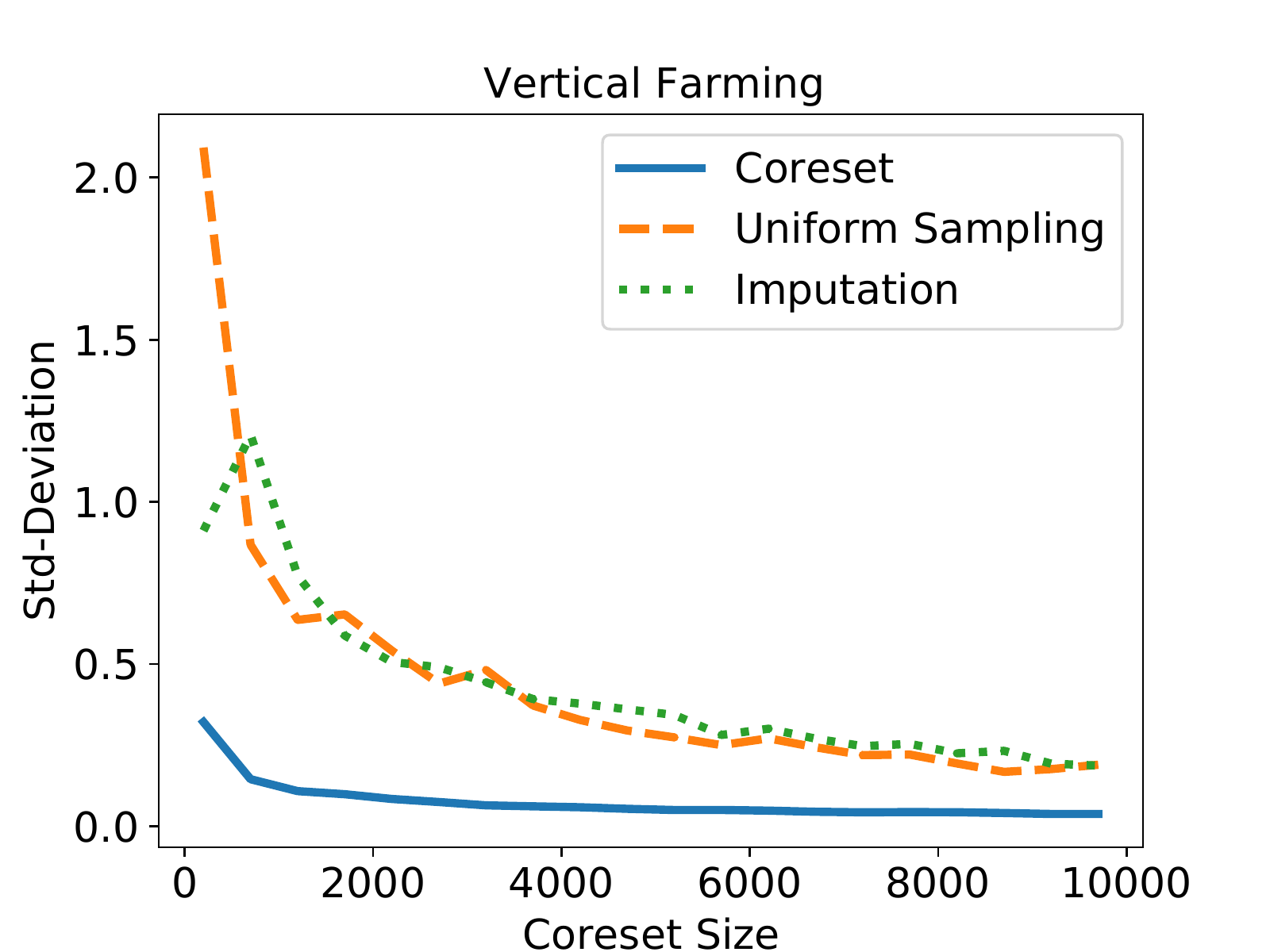}
        \caption{}
    \end{subfigure}
    \begin{subfigure}[b]{0.4\textwidth}
        \centering
        \includegraphics[width=\textwidth]{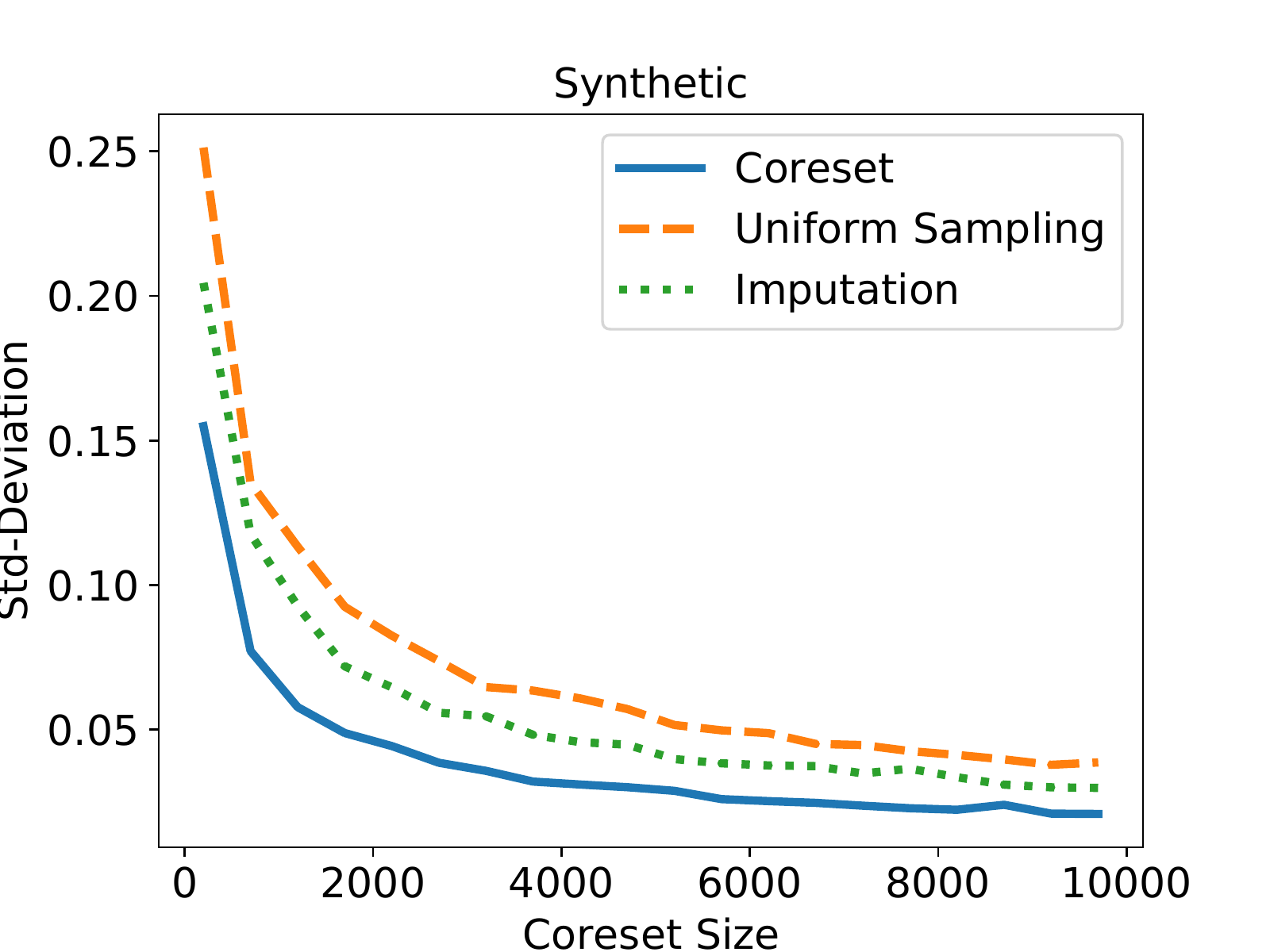}
        \caption{}
    \end{subfigure}
    \caption{Standard deviation for the size-error evaluation.}
    \label{fig:std}
\end{figure}

\subsection{Speedup of Lloyd's-style Heuristic}
Coresets often help to speed up existing approximation algorithms.
Before our work, the only algorithm for \kMeans with provable guarantees for multiple missing values was~\cite{DBLP:conf/soda/EibenFGLPS21}. Unfortunately, \cite{DBLP:conf/soda/EibenFGLPS21} is not practical even when combined with coresets, since it contains several enumeration procedures that require $\Theta(\exp(\poly(\epsilon^{-1}jk)))$ time.
We consider a variant of Lloyd's heuristic~\cite{DBLP:journals/tit/Lloyd82} that is adapted to the missing-value setting,
and we evaluate its speedup with coresets.
The algorithm is essentially the same as the original Lloyd's algorithm,
except that the distance as well as the optimal $1$-mean for a cluster (which can be computed optimally in $O(d |P|)$ for a cluster $P$~\cite{DBLP:conf/soda/EibenFGLPS21}),
is computed differently.
We show that our coreset can significantly accelerate this algorithm.
In particular, we run the modified Lloyd's heuristic
directly on the original dataset,
and take its running time and objective value as the comparison reference.
Then we run this modified Lloyd's heuristic again, but
on top of our coreset and the uniform sampling baseline respectively,
and we compare both the speedup and the relative error\footnote{For $x \in \mathbb{R}_+$, the relative error of $x$
against a reference $x^\star > 0$ is defined as $\frac{|x - x^\star|}{x^\star}$.} against the reference.
The experiments are run on the Russian housing data set where the number of iterations of the modified Lloyd's is set to $T=5000$ and the number of clusters is set to a small value $k=3$ so as the heuristic is likely to find a local minimum faster.
Again, to obtain a stable result, we run the experiments for $40$ times with independent random bits and report the average relative errors and running time.

\paragraph{Results.}
The relative error with respect to varying coreset sizes can be found in Figure~\ref{figs:lloyd_err}.
We can see that the relative error of Lloyd's algorithm running on
our coreset is consistently low,
while the uniform sampling baseline has several times higher error and the error does not seem to improve even when improving the size.
We note that relative errors for both our coreset and uniform sampling
are significantly lower than that we observe from the empirical error in Figure~\ref{fig:rh_accu}.
In fact, they are not necessarily comparable since the empirical error in Figure~\ref{fig:rh_accu} is always evaluated on a same center,
while what we compare in Figure~\ref{figs:lloyd_err} is the center sets found by the modified Lloyd's running
on different data sets.
This also helps to explain why improving the size of uniform sampling may not result in a better solution,
since as shown in Figure~\ref{fig:rh_accu}, uniform sampling has
a large empirical error (around $50\%$), so a good solution for the uniform sample
may not be a good solution for the original data set.

The running time of the modified Lloyd's on top of our coresets can be found in Figure~\ref{figs:lloyd_time},
and the running time of Lloyd's on the original dataset is $22.9$s (which is not drawn on the figure).
To make a fair comparison, we also take the coreset construction time into account.
Note that coreset size is not a dominating factor in the running time of coreset construction,
since the majority of time is spent on computing the importance scores and the coreset size only affects the number of samples.
A coreset of size only $1000$
can achieve $<1\%$ error, and the running time of constructing the coreset and
applying Lloyd's on top of it are $3$s and $0.8$s, respectively,
which offers more than $5$ times of speedup.
We remark that our experiments only demonstrate the speedup in a single-machine scenario,
and the speedup will increase in the parallel or distributed setting.

\begin{figure}[t]
    \centering
    \captionsetup{font=small}
    \begin{subfigure}[b]{0.43\textwidth}
        \centering
        \includegraphics[width=\textwidth]{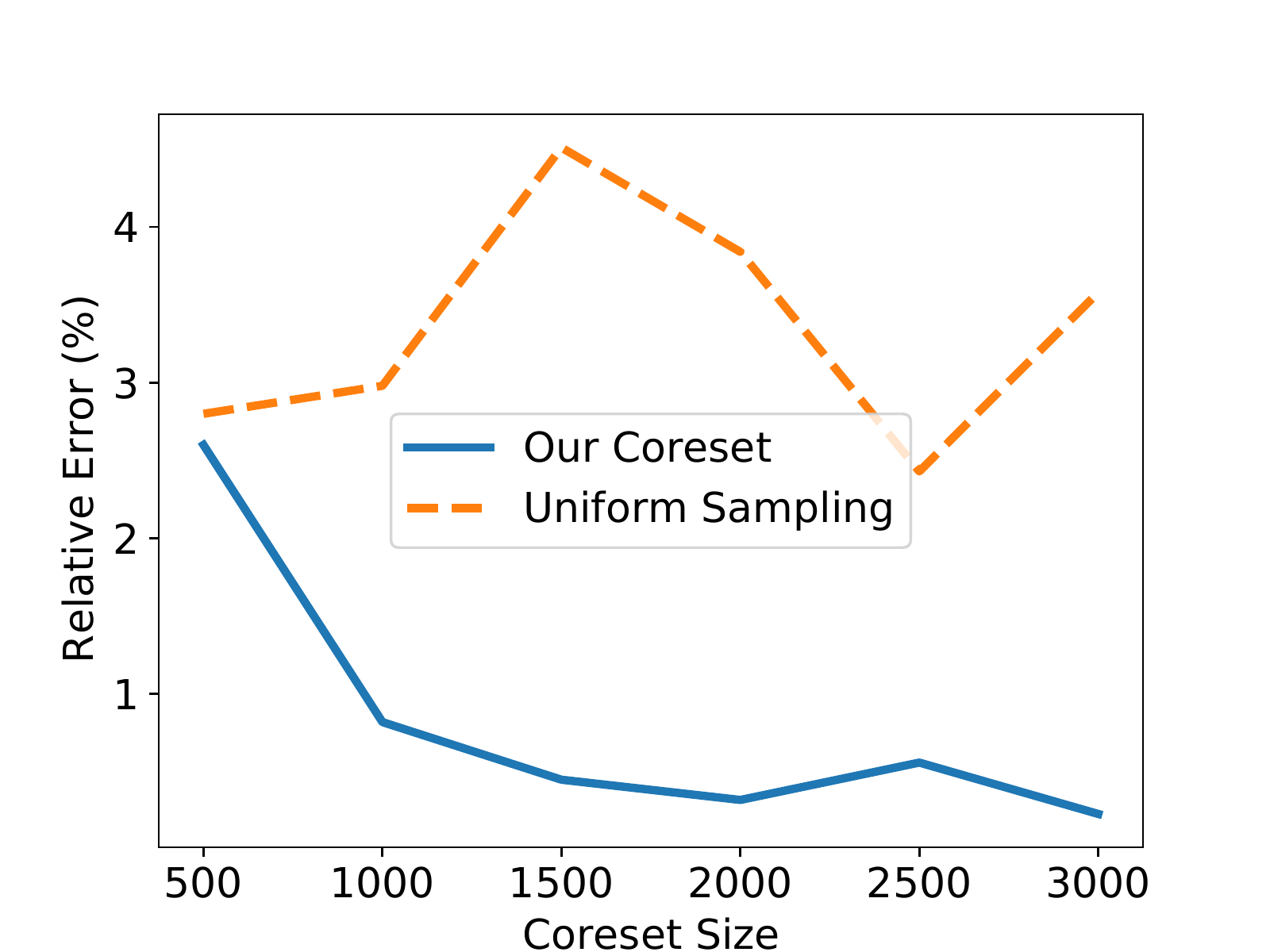}
        \caption{} \label{figs:lloyd_err}
    \end{subfigure}
    \begin{subfigure}[b]{0.43\textwidth}
        \centering
        \includegraphics[width=\textwidth]{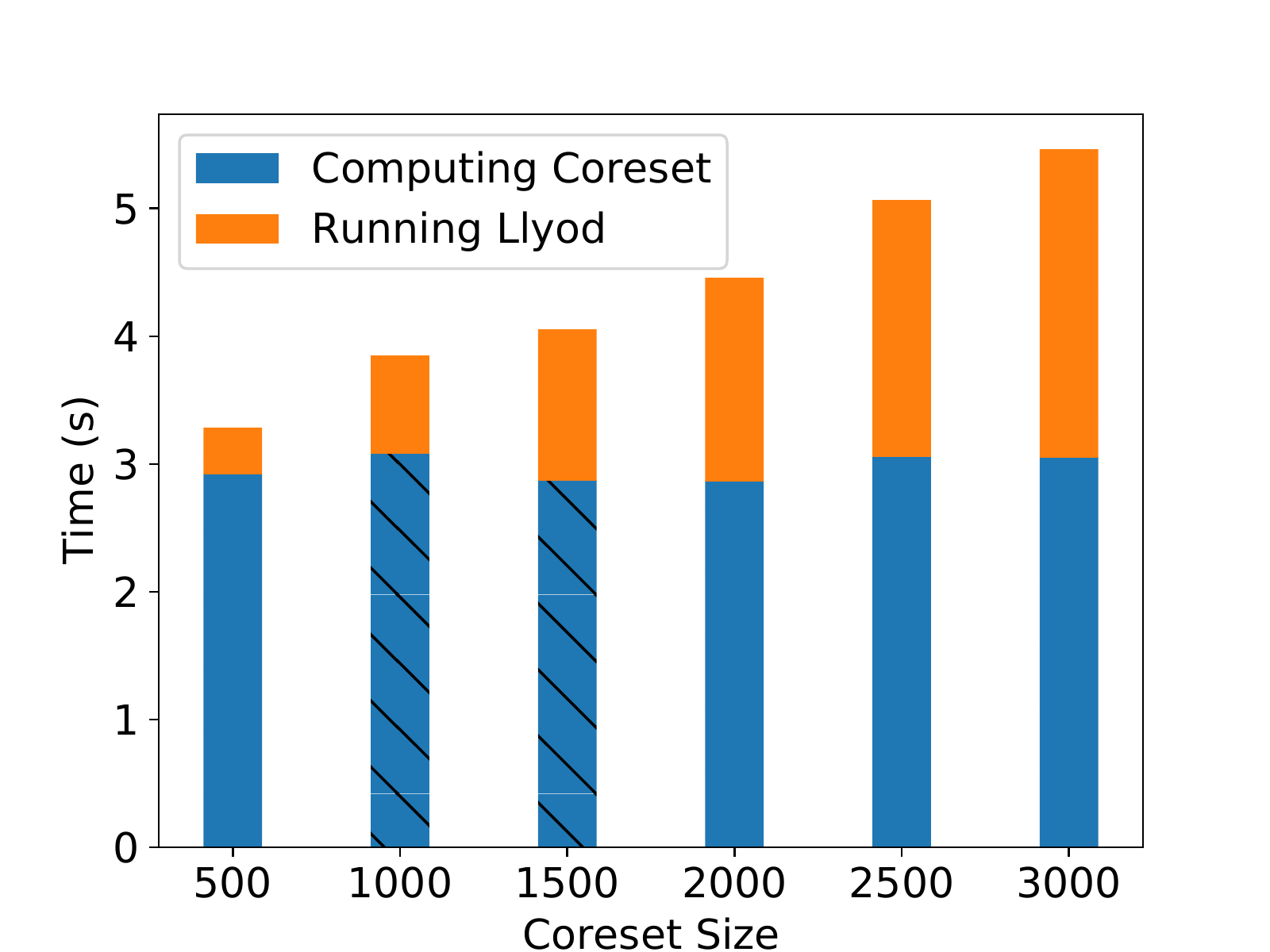}
        \caption{} \label{figs:lloyd_time}
    \end{subfigure}
    \caption{Relative error and running time evaluation for the Lloyd's heuristic on the coreset, with respect to varying coreset sizes.
    The left figure demonstrates the relative error,
    and the right figure shows the running time of constructing our coreset,
    and the time for the modified Lloyd's heuristic running on top of our coreset.}
    \label{fig:speedup}
\end{figure}

 \section{Lower Bounds}
We prove the following lower bound to assert the necessity of the exponential dependence on $\min(j,k)$ in our coreset construction Theorem~\ref{thm:main_full}.

\begin{theorem}[Restatement of Theorem~\ref{thm:lb}]
    \label{thm:lb_restate}
    Consider the \kMeans with missing values problem in $\RR^d_{?}$ where each point can have at most $j$ missing coordinates. Assume there is an algorithm that constructs an $\epsilon$-coreset of size $f(j,k)\cdot \poly(\epsilon^{-1}d\log n)$, then $f(j,k)$ can not be as small as $2^{o(\min(j,k))}$.
\end{theorem}

\begin{proof}
Consider the following $n$ points instance with $j=k=\Theta(\log n)$, and $d=2j$. For a subset $I$ of $[d]$, we define a data point $p(I)$ such that $p(I)_i=1$ if $i\in I$ and $p(I)_i=?$ otherwise.
Then we let the data set $P=\{p(I)|I\subseteq [d],|I|=j\}$. We remark that we can make $|P|=\binom{d}{j}=n$ by choosing a proper $j=\Theta(\log n)$.

We prove that any $1/2$-coreset of $P$ should contain every point in $P$. Let $D$ be such a coreset and assume $p(I)\not\in D$, we choose the following $k=j$ centers. For every $i\in I$, we define a center $c^i\in \mathbb{R}^d$ such that the $i$-th coordinate of $c^i$ is $0$ and the other coordinates of $c^i$ are $1$. We observe that, for any $i\in I$, $\mathrm{dist}(p(I),c^i)=1$. Meanwhile for any other $p(I')\not=p(I)$, there must be a $i'\in I\setminus I'$ since $|I|=|I'|$, thus $\mathrm{dist}(p(I'),c^{i'})=0$. This should imply that the cost on coreset is $0$ while the cost on $P$ is $1$ which makes a contradiction.

Since $j=k=\Theta(\log n)$, $d=2j$, we have $2^{o(\min(j,k))}\cdot \poly(d\log n)=o(n)$. Thus $f(j,k)$ can not be as small as $2^{o(\min(j,k))}$.
\end{proof}  \section{Conclusion}
Our coreset construction builds upon the sensitivity-sampling method (cf.~\cite{DBLP:conf/stoc/FeldmanL11}).
However, a central technical challenge is that the standard method to compute the sensitivity scores breaks,
because distances between points with missing values do not satisfy the triangle inequality.
We overcome this using another known method, of~\cite{varadarajan2012near}, that requires a coreset for \kCenter.
Our main innovation is a near-linear time algorithm that computes an
$O(1)$-approximate \kCenter coreset for points with missing values.
To this end, we need the following key steps, which constitute our main technical contribution. 
\begin{itemize}
    \item We reduce the \kCenter coreset construction with missing values,
    to the construction of traditional \kCenter coresets (i.e., without missing values) on a series of instances.
    These instances are built by restricting data points with missing values
    to a carefully-chosen collection of subspaces.
    The guarantee needed from this collection is a certain combinatorial structure, and we indeed prove it exists. 
    \item The method of Varadarajan and Xiao executes
    the \kCenter coreset algorithm many times, and overall takes quadratic time.
    To improve the running time, we design an efficient dynamic algorithm
    for the well-known Gonzales’ algorithm
    (which computes an $O(1)$-approximate \kCenter coreset).
    The main idea in this dynamic algorithm is to project the data points onto (data-oblivious) random 1D lines, and build on each line a dynamic data structure that supports furthest-neighbor queries (in 1D).
\end{itemize}
Finally, we implemented our algorithm and the experiments indicate that our algorithm is efficient and accurate enough to be potentially applicable in practice.

\paragraph{Future directions.}
As an immediate follow-up, one could try to improve our coreset size,
e.g., removing the dependence in $\log n$.
Our input can be viewed as axis-parallel affine-subspaces.
Hence, another an interesting direction is to obtain coresets for the more general setting where the input consists of general affine-subspaces.

\paragraph{Potential negative societal impacts.}
Our paper focuses on computational issues (improving time and space) of known clustering tasks.
Clustering methods in general have potential issues 
with fairness and privacy, which applies also to our work,
but our research is not expected to introduce new negative societal impact
beyond what is already known.  
\begin{ack}
  The majority of this work was done when Shaofeng Jiang was at Aalto University.
  This work is partially supported by ONR Award N00014-18-1-2364,
    by the Israel Science Foundation grant \#1086/18,
    by a Weizmann-UK Making Connections Grant,
    and by a Minerva Foundation grant.
\end{ack}

\bibliographystyle{alphaurl}
\bibliography{refs}

\newpage

\end{document}